\documentclass[11pt]{article}
\usepackage[utf8]{inputenc}
\usepackage[colorlinks=true,linkcolor=green,citecolor=blue,allcolors=green,hypertexnames=false]{hyperref}
\hypersetup{
   colorlinks=true,
   linkcolor=blue,
   filecolor=blue,
   urlcolor=blue,
   allcolors=blue
}
\usepackage{fullpage}
\usepackage{amsfonts,amsmath}
\usepackage{graphicx,xcolor}
\usepackage{paralist}
\usepackage{appendix}
\usepackage{algpseudocode}
\usepackage{multirow}
\usepackage{multicol}
\usepackage{amsthm}

 \usepackage{bbm}

\newcommand{\old}[1]{\textcolor{gray}{OLD: #1}}

\newcommand{\odelta}{\overline{\delta}}

\renewcommand{\th}{^\textrm{th}}
\newcommand{\comment}[1]{\textcolor{black!70}{\textsf{** #1 **}}}

\newcommand{\om}{\widetilde{m}}
\newcommand{\hm}{\hat{m}}

\newcommand{\eps}{\epsilon}

\newcommand{\poly}{\textrm{poly}}

\newcommand{\Ex}{\mathbb{E}}
\newcommand{\EX}{\mathbb{E}}

\newcommand{\eqdef}{\stackrel{\rm def}{=}}

\newcommand{\ot}{\widetilde{t}\,}
\newcommand{\hatt}{\hat{t}}

\newcommand{\calA}{{\mathcal{A}}}
\newcommand{\mA}{{\mathcal{A}}}

\newcommand{\mG}{{\mathcal{G}}}

\newcommand{\badadvice}{\textsf{``bad advice"}}

\newcommand{\settauD}{\frac{8m\gamma^2}{\eps\cdot \ot}}
\newcommand{\settauT}{12\gamma/\eps}
\newcommand{\setQ}{\frac{m}{(\eps \ot)^{2/3}}\cdot 30\ln(4/\delta)}
\newcommand{\setq}{\frac{16m\cdot \tau_t\cdot \ln(4/\delta)}{\eps^2\cdot \ot}}

\newcommand{\vvv}{{v}}
\newcommand{\vecV}{\vec{P}}
\newcommand{\hv}{\hat{v}}
\newcommand{\ov}{\overline{v}}
\newcommand{\ert}[1]{E_{rt}[#1]}
\newcommand{\invokeA}{\calA}
\newcommand{\search}{\textsf{Search}}
\newcommand{\oa}{\overline{\alpha}}
\newcommand{\ta}{\widetilde{\alpha}}
\newcommand{\aG}{\alpha(G)}

\renewcommand{\hatt}{\hat{t}}

        {\hspace*{\fill}$\Box$\par}

\newtheorem{theorem}{Theorem}[section]
\newtheorem{corollary}[theorem]{Corollary}
\newtheorem{lemma}{Lemma}[section]
\newtheorem{claim}[lemma]{Claim}
\newtheorem{definition}[lemma]{Definition}
\newtheorem{observation}[lemma]{Observation}
\newtheorem{notation}[lemma]{Notation}
\newtheorem{remark}[lemma]{Remark}

\makeatletter
\def\moverlay{\mathpalette\mov@rlay}
\def\mov@rlay#1#2{\leavevmode\vtop{%
   \baselineskip\z@skip \lineskiplimit-\maxdimen
   \ialign{\hfil$\m@th#1##$\hfil\cr#2\crcr}}}
\newcommand{\charfusion}[3][\mathord]{
    #1{\ifx#1\mathop\vphantom{#2}\fi
        \mathpalette\mov@rlay{#2\cr#3}
      }
    \ifx#1\mathop\expandafter\displaylimits\fi}
\makeatother

  \newcommand{\isHeavy}{\hyperref[alg:IsHeavy]{\color{black}{\sf IsHeavy}}}

\newcommand{\IsAssigned}{\hyperref[alg:IsAssigned]{\color{black}{\sf IsAssigned}}}

\newcommand{\TrianglesApprox}{\hyperref[alg:triangles]{\color{black}{\sf Approx-Triangles-With-Arboricity-Advice}}}

\newcommand{\ApproxTriangles}{\hyperref[alg:triangles]{\color{black}{\sf Approx-Triangles-With-Arboricity-Advice}}}

 \newcommand{\approxEdges}{\hyperref[alg:edges]{\color{black}{\sf Approx-Edges-With-Arboricity-Advice}}}

\newcommand{\alg}[2]{
 \begin{figure*}[htb!]
	\centering
       \fbox{\parbox{0.95\linewidth}{
	#2
        }}
        \label{#1}
\end{figure*}
}

\begin{document}

\title{Testable algorithms for approximately counting edges and triangles in sublinear time and space}
\author{Talya Eden\thanks{Bar-Ilan University,  talyaa01@gmail.com. 
Talya gratefully acknowledges the support of the MIT International Science and Technology Initiatives (MISTI) and the MIT-Israel Zuckerman STEM Fund. 
Part of this work was conducted while the authors were visiting the Simons
Institute for the Theory of Computing. 
} 
\and
Ronitt Rubinfeld\thanks{Massachusetts Institute of Technology, ronitt@mit.edu. 
Supported by NSF awards DMS-2022448 and CCF-2310818,
and the MIT-Israel Zuckerman STEM Fund. Part of this work was conducted while the authors were visiting the Simons
Institute for the Theory of Computing.
} 
\and
Arsen Vasilyan\thanks{University of Texas at Austin. Supported by NSF AI Institute for Foundations of Machine
Learning (IFML) and NSF award
CCF-2310818. 
Arsen gratefully acknowledges the support of the MIT International Science and Technology Initiatives (MISTI) and the MIT-Israel Zuckerman STEM Fund.
Part of this work was conducted while the authors were visiting the Simons
Institute for the Theory of Computing. 
}
}

%

\maketitle

\date{}

\begin{abstract}

We consider the fundamental
problems of 
approximately counting the numbers of edges and
triangles in a graph in sublinear time.
Previous algorithms for these tasks 
are significantly more efficient under a promise that the arboricity of the graph is bounded by some parameter $\oa$. However, when this promise is violated, the estimates given by these algorithms are no longer guaranteed to be correct.

For the triangle counting task, we give an algorithm that requires no promise on the input graph $G$, and computes a $(1\pm \epsilon)$-approximation for the number of triangles $t$ in $G$ in time $O^*\left(
\frac{m\cdot \aG}{t}
+
\frac{m}{t^{2/3}}
 \right)$, where $\aG$ is the arboricity of the graph. The algorithm can be used on any graph $G$ (no prior knowledge the arboricity $\aG$ is required), and the algorithm adapts its run-time on the fly based on the graph $G$.

We accomplish this by trying a sequence of candidate values $\ta$ for $\aG$ and using a novel algorithm in the framework of testable algorithms.
This ensures that wrong candidates $\ta$ cannot lead to incorrect estimates: as long as the advice is incorrect, the algorithm detects it and continues with a new candidate.
Once the algorithm accepts the candidate, its output is guaranteed to be correct with high probability. 

We prove that this approach preserves - up to an additive overhead - the dramatic efficiency gains obtainable when good arboricity bounds are known in advance, while ensuring robustness against misleading advice. We further complement this result with a lower bound, showing that such an overhead is unavoidable whenever the advice may be faulty.

We further demonstrate implications of our results for triangle counting in the streaming model.

\end{abstract}
\newpage
\tableofcontents

\newpage
\section{Introduction}
\vspace{-0.3em}
Counting triangles is a fundamental graph problem with applications in network analysis, database systems, and social network theory, among others (see
\cite{al2018triangle} for a comprehensive survey).
A multitude of algorithms for counting triangles have been developed 
which 
work well in sequential~\cite{nesetril,ChNi85,itai1977finding}, streaming~\cite{kallaugher2019complexity,mcgregor2016better,chen2022triangle},
sublinear time~\cite{ELRS,ERS_cliques,BiswasER21,hamiltonian,ERS19_dd,AKK,Peng}, parallel and distributed models~\cite{Biswas_MPC_triangles,fischer2018possibilities,censor2022deterministic,czumaj2020detecting,pagh2012colorful,KoPiPlSe13,park2013efficient},
both from 
theoretical and practical perspectives.
The best known
\emph{sublinear-time} algorithm for
estimating the number of triangles
in general graphs runs in 
$O^*(\frac{m^{3/2}}{t})$ time~\cite{ELRS}, where  $m$ and $t$ denote the number of edges and triangles in the input graph $G$, respectively.
This algorithm is 
known to be optimal due to query complexity lower bounds (in the augmented\footnote{In the augmented query access model it is assumed that the graph has access to uniform vertex samples, degree queries, uniform neighbor queries, pair queries, as well as  uniform edge samples.} query access model).

Given the importance of the triangle counting problem, 
there is much interest in even faster algorithms for
typical real world graphs.   
These graphs often have structural properties that make
it possible to  perform faster computations.
An important example of such
a property is the graph’s \emph{arboricity}, 
a sparsity parameter
that corresponds (up to a constant  factor of 2) 
to the maximum average degree over all subgraphs. 
The arboricity
is typically small in many real-world networks~\cite{goel2006bounded,JS17,shin2018patterns,danisch2018listing}, 
including planar graphs and preferential attachment networks. 
When an upper bound on the arboricity is
known, there are fast sublinear-time triangle-counting algorithms 
\cite{ERS_faster, bishnu2025arboricity}
that circumvent
the 
known lower bounds for general graphs.


The aforementioned arboricity-based algorithms require an upper bound $\oa$ on the actual arboricity of the graph, denoted $\aG$,
as one of their input parameters. Those algorithms
rely crucially on the correctness of this assumption. 
If the arboricity of the graph exceeds $\oa$, 
these algorithms are no longer guaranteed to give correct estimates for the number of triangles.
\vspace{-1em}
\subsection{Our Contribution}
In this work, we develop a novel algorithm in the augmented query access model that does not require any information about the arboricity of the input graph.
For \emph{all graphs} $G$, our algorithm gives a correct $(1\pm \epsilon)$-approximation to the number of triangles in $G$, and has a run-time that scales with the arboricity of the input graph. Specifically, the run-time of our algorithm is $O^*\left(
\frac{m\cdot 
\aG}{t}
+
\frac{m}{t^{2/3}}
 \right)$.\footnote{We use the $O^*(\cdot)$ notation to hide $\poly(1/\eps, \log n)$ dependencies.}

A naive approach would be to 
combine the sublinear-time arboricity estimator of~\cite{EstArb} with 
the triangle counting results of~\cite{ERS_faster, bishnu2025arboricity}. However, \cite{EstArb} runs in 
essentially optimal time $\Theta(n/\aG)$, and so 
invoking it would dominate the complexity and substantially degrade the overall runtime, reflecting the inherently high query cost of arboricity estimation.

We avoid estimating arboricity directly, by first designing an algorithm in the formal framework of 
testable algorithms, introduced by~\cite{RubinfeldV23} in the context of agnostic learning. A testable algorithm receives advice about the input
that may or may not be accurate. The algorithm is allowed to either output an estimate or terminate early  indicating $\badadvice$.  
It must satisfy the following two properties: 
(1) {\em Completeness}: If the advice is correct, the algorithm outputs a correct estimate (with high probability).
(2) {\em Soundness}: If the algorithm outputs an estimate (i.e., does not return $\badadvice$), then that estimate is correct (with high probability), even if the advice was incorrect.
This framework captures natural scenarios in which domain knowledge provides partial information about the input class, but the reliability of this information cannot be guaranteed.

We note that previous work on testable algorithms focused on settings in learning theory, and our central conceptual contribution is extending this framework to approximate counting and estimation tasks. 
Specifically, we design efficient testable algorithms for estimating the number of 
edges and triangles in a graph, given a potentially unreliable upper bound $\ta$ on its arboricity as advice.\footnote{Throughout the paper we use $\aG$ to denote the actual arboricity of the graph $G$, $\oa$ to denote a guaranteed upper bound on the arboricity, and $\ta$ to denote an unreliable upper bound.}
Our algorithms  return a $(1 \pm \varepsilon)$-approximation with high probability, provided it does not output $\badadvice$. 
Moreover, if the given advice is accurate, then the algorithm is not likely to output $\badadvice$.


\begin{theorem}[Testable approximate triangle counting algorithm]
\label{thm: main theorem after search for ot}
Let $G$ be the input
graph, and $\eps$and $\delta$ approximation and error parameters.
Let $m$ and $t$ be the number of edges and triangles in $G$, respectively. 
There exists an algorithm that
gets as input
augmented query
access to $G$, the number\footnote{ We remark that the requirement on knowing
$m$ can be removed at an additive (essential) overhead of 
$O^*(\frac{n\aG}{m})$, as implied by Corollary~\ref{cor:edges} below or $O^*(n^{1/4})$ due to~\cite{beretta2025faster}. However, for the sake of clarity, we chose to present the two results separately and focus here on the triangle estimation problem, as this is the main focus of our work. } of edges $m$, 
as well as 
an \emph{advice:}  a positive integer $\ta$.
The algorithm has
expected run-time and query complexity of $O^*\left(\left(\frac{m \ta }{ t}+\frac{m }{ t^{2/3}} \right)\log \frac{1}{\delta}\right)$. Upon each invocation, the algorithm either returns \badadvice{} or a number $\hat{t}$, satisfying
\begin{itemize}
\item  (Completeness)
For every $G$, if $\aG\leq  \ta$,
then the algorithm can output $\badadvice$ only with probability at most $\delta$.

\item  (Soundness)
For every $G$ and $
\ta$,  w.p. at least $1-\delta$ the algorithm will either return $\badadvice$ or output $\hat{t} \in (1 \pm \epsilon)t$.
\end{itemize}
\end{theorem}

In order to eliminate the need for advice, we use this algorithm 
with a sequence of guesses serving as our advice value $\ta$. We start with running the algorithm with a setting of $\ta=1$, and continue invoking the algorithm while doubling 
the guessed value of $\ta$ until the algorithm outputs an answer (rather than $\badadvice$). 
By Soundness, the first accepting run outputs a correct estimate with high probability, and by Completeness, the number of guesses is bounded, so the overhead is only logarithmic. Consequently, we obtain an algorithm applicable to any graph $G$.
%

 \begin{corollary}[Instance-adaptive approximate triangle-counting algorithm.]
\label{cor: triangle-counting algorithm that adapts to arboricity}
    There exists an algorithm with the following guarantees. The algorithm receives augmented query access to a graph $G$, and the number of edges $m$ in $G$ together with  parameters $\epsilon,\delta\in(0,1)$. With probability at least $1-\delta$, the algorithm produces an estimate $\hat{t}$ that with probability at least $1-\delta$ satisfies $\hat{t}\in (1\pm \epsilon)t$, where $t$ is the number of triangles in $G$. The expected run-time of the algorithm is $O^*\left(\left(\frac{m\cdot \aG}{t}
+
\frac{m }{t^{2/3}}
 \right) \log\left(\frac{1}{\delta}\right)\right)$, where $\aG$ is the arboricity of $G$.    
\end{corollary}

We refer to this algorithm as an \emph{instance-adaptive algorithm} since it adapts ``on the fly" to the structure of 
$G$, scaling with 
$\aG$ without requiring any a priori bound on it.

\vspace{-1em}
\paragraph{Optimality of the bound.}
The query complexity of our testable algorithm for
approximating the number of triangles lies  between  two extremes: the known lower bound in the general setting (with no advice), 
and the more efficient algorithms achievable under a reliable arboricity promise. 
We prove that this tradeoff is inherent: any algorithm that uses significantly fewer queries may produce an underestimate of the triangle count when the advice is incorrect, thus establishing the optimality of our approach. The lower bound also establishes a strict separation between the instance-adaptive estimation setting we consider here, and the trusted-advice setting of \cite{ERS_faster}.
Formally, 
we give an information-theoretic lower-bound showing that $\Omega\left(
\frac{m}{t^{2/3}}
 \right)$ queries are necessary,  which together with the 
 $\Omega(\frac{m\cdot \aG}{t})$ lower bound of~\cite{ERS_faster}
 demonstrates the optimality of our algorithm.

\begin{theorem}\label{thm:lb_triangle_counting}
    Let $\mA$ be a testable algorithm for approximating the number of triangles in a graph $G$ given arboricity advice $\ta$ and augmented query access to $G$. 
    Let $n,m$ and $t$ denote the number of nodes, edges and triangles in $G$. Finally, assume that the algorithm knows $n$ and $m$, and that $\ta\geq m/n$.
    Then
    $\mA$ must perform $\Omega(\frac{m}{t^{2/3}}+\frac{m\cdot \aG}{t})$ queries in order to succeed with high probability. 
\end{theorem}


\paragraph{Further results.}

Using a known reduction from sublinear-time algorithms in the augmented model to streaming algorithms in the edge arrival model (see, e.g., ~\cite{EstArb,fichtenberger2022approximately}), 
we prove the following.
\begin{theorem}
\label{cor: streaming}
    There exists a $9$-pass $O^*\left(\frac{m\cdot \ta \log 1/\delta}{t}
+
\frac{m \log 1/\delta}{t^{2/3}}
 \right)$-space algorithm for testable triangle-counting with arboricity advice (Definition \ref{def:motif_counting_algorithm}) in the arbitrary-order insert-only edge arrival streaming model. As is standard and unavoidable in the streaming model, the algorithm requires a rough initial estimate $\ot = \Theta(t)$ of the triangle-count $t$.\footnote{In our implementation we assume $\ot\in[t/4,t]$, but the same argument can be adapted to work with any given $\ot \in [\frac{1}{c}\cdot t,c\cdot t]$ for any known constant $c$.}
\end{theorem}

In the guaranteed advice setting, Bera and Seshadhri~\cite{bera2020degeneracy} gave an arboricity-dependent 6-pass $O^*(m\oa/t)$-space  streaming algorithm for triangle counting.  

We note that  the instance-adaptive version of our algorithm can also be implemented in the streaming model, however with a $O(\log n)$ round complexity, due to the need to search for the correct advice $\oa$. For further discussion see Section~\ref{sec: streaming}.
We also note that lower bounds do not directly transfer between the two models, and so it remains open as to
whether
the space bounds and round complexity of Theorem~\ref{thm: main theorem after search for ot}  are optimal.


Finally, we state our result for edge estimation in the testable setting.

\begin{theorem}\label{thm: testable edge estimation theorem}

    There exists a testable algorithm in the augmented query model (Definition \ref{def:motif_counting_algorithm}) for $(1\pm \epsilon)$-approximate counting of edges on a graph $G$ with $n$ vertices, $m$ edges and $t$ triangles, and has expected run-time of $O^*\left(
\frac{n\cdot \oa \log 1/\delta }{m}
 \right)$.
\end{theorem}


Interestingly, this result does not incur any additive overhead over the $O^*(\frac{n\oa}{m})$ results of~\cite{ERS19_dd}. It does, however, 
require the  uniform edge samples of the stronger augmented query access model, 
as opposed to the algorithm of \cite{ERS19_dd}.
This stronger query access is necessary if one wants to obtain an arboricity-dependent algorithm in the  testable setting, as otherwise 
the presence of a small (hard to find)
clique of  size  $\sqrt{m}$ can fool any algorithm
which uses at most $o(n/\sqrt{m})$ vertex samples. For small-arboricity graphs, the resulting algorithm is much faster than the state-of-the-art edge-counting algorithm for general graphs  in the augmented model
~\cite{beretta2025faster}.
In particular, for constant arboricity graphs, the running time of our algorithm is $O^*(1)$, whereas the 
general algorithm of~\cite{beretta2025faster}, which does not rely on arboricity,  runs in 
essentially optimal time $\Theta^*(n^{1/4})$.

As in the case of triangle counting,
an easy consequence of our testable algorithms is an instance-adaptive algorithm for approximating the number of edges (using a similar doubling trick).
 Formally, the resulting instance-adaptive approximate edge counting algorithm runs in time $O^*\left(
\frac{n\cdot 
\aG}{m}
 \right)$.

\medskip
 Our results are summarized below.


\begin{table}[h]
\label{fig: table of results}
\centering
\renewcommand{\arraystretch}{1.8}
\begin{tabular}{|c|c|c|c|c|}
\hline
  & Computational model
  & General Graphs
  & \shortstack{Guaranteed Bound\\on Arboricity}
  & \shortstack{Our Results\\(Testable Algorithms)} \\
\hline
\multirow{2}{*}{Triangles}
  & \shortstack{sublinear-time,\\augmented}
  & $\Theta^*\!\left(\tfrac{m^{3/2}}{t}\right)$
  & \multirow{2}{*}{$\Theta^*\!\left(\tfrac{m\oa}{t}\right)$}
  & $\Theta^*\!\left(\tfrac{m\cdot \aG}{t} + \tfrac{m}{t^{2/3}}\right)$ \\
\cline{2-3} \cline{5-5}
  & \shortstack{streaming,\\edge arrival}
  & $\Theta^*\!\left(\min\!\left\{\tfrac{m^{3/2}}{t}, \tfrac{m}{\sqrt{t}}\right\}\right)$
  &
  & $O^*\!\left(\tfrac{m\ta}{t} + \tfrac{m}{t^{2/3}}\right)$ \\
\hline
Edges
  & augmented model
  & $\Theta^*\!\left(n^{1/4}\right)$
  & $\Theta^*\!\left(\tfrac{n\oa}{m}\right)$
  & $\Theta^*\!\left(\tfrac{n\cdot \aG}{m}\right)$ \\
\hline
\end{tabular}

\caption{Summary of our results, compared to the state of the art in the same model for general graphs as well as bounded arboricity graphs. 
Our results can be seen as lying in between these two ``extremes".
Here (and elsewhere) $\aG$ denotes the actual arboricity of the graph, $\oa$ a guaranteed upper bound on $\aG$, and $\ta$ a potentially-incorrect advice on $\aG$.  }
\end{table}

\subsection{Overview of the algorithms and lower bounds for triangle estimation}

\subsubsection{Testable triangles algorithm}

\textbf{The ERS algorithm.} The algorithm by~\cite{ERS_faster}, henceforth ERS,  
begins by assigning triangles to one of their \emph{light} edges, defined with respect to a \emph{triangle-heaviness threshold $\tau_t$}. 
Triangles for which all three edges are heavy -- \emph{heavy triangles} -- remain unassigned. 
This assignment rule yields a low-variance, 
nearly unbiased estimator. Consequently, 
sampling a set $R$ of $O(m\tau_t/t)$ edges uniformly at random ensures that, 
with high probability, the number of triangles assigned to some edge in $R$ is close to its expected value. To estimate this quantity, the algorithm repeatedly tries to sample triangles incident to edges of $R$, and whenever a triangle incident to some edge $e\in R$ is found, it checks whether the triangle is indeed assigned to $e$.

ERS prove that estimating the number of triangles assigned to \(R\) requires only \(O(d(E)/t)\) triangle-sampling attempts in expectation, where \(d(E)\eqdef\sum_{e\in E} d(e)\),  and $d(e)$ for $e=\{u,v\}$ is defined as $\min\{d(u),d(v)\}$.
Once a triangle is found, verifying whether it is assigned to an edge $e$ can be done in $O(d(e)/\tau_t)$ queries. To bound this term, they  define a \emph{degree-heaviness threshold $\tau_d$}
beyond which edges are also classified as heavy (independent of their triangle count) and
then ignored. Hence, the cost of an  assignment verification step is $O(\tau_d/\tau_t)$.

The arboricity guarantee then allows them to bound each of these three terms by $O(m\oa/t)$, through the following arguments:
\begin{enumerate}[(i)]
\item \textbf{Triangle-heaviness:} 
In bounded-arboricity graphs,
edges with triangle load above $\tau_t=\oa/\eps$ can be ignored, since 
they can only contribute  a small fraction of triangles. 
This determines the sample size of $R$.

\item \textbf{Sum of edge degrees:} In bounded-arboricity graphs,
the sum of edge degrees satisfies $d(E)=\sum_{e\in E}d(e)=O(m\oa)$, 
implying that the number of triangle sampling attempts is bounded.

\item \textbf{Degree-heaviness:} In bounded-arboricity graphs,
only few edges exceed $\tau_d=m\oa^2/t$, 
so their contribution to the total triangles count is small and assignment verification remains efficient.
\end{enumerate}

Plugging in each of these bounds in their respective terms gives the $O(m\oa/t)$ running time.

\medskip
\textbf{The testable algorithm.} 
Our testable algorithm uses the three properties above to verify the advice $\ta$. If these properties hold, then the 
correctness does not require that the actual arboricity be small; 
rather, these properties suffice to guarantee the quality of the estimate.

However, since the advice is not guaranteed,  setting the thresholds $\tau_t$ and $\tau_d$ as in ERS no longer ensures that the number of unassigned heavy triangles is small. We must therefore modify the thresholds and explicitly check that the number of heavy edges  remains bounded. In ERS, the set of heavy edges $H$ satisfies $|H|=O(\eps t/\oa)$, and arboricity is used to argue that $H$ induces at most $O(\oa|H|)=O(\eps t)$ triangles. Without the arboricity guarantee, we instead rely on the general bound that any $m$-edge graph contains at most $O(m^{3/2})$ triangles. Hence we would like to make sure that $|H|=O((\eps t)^{2/3})$.
 Accordingly, we set 
\[\gamma=c\cdot \max\{\ta,t^{1/3}\}, \;\;\; 
\tau_t=\gamma/\eps, \;\;\;  \text{and} \;\; \tau_d=m\gamma^2/t,\]
for some small constant $c$. 
We show that these settings ensure that if the advice is correct then $|H|$ will be small. Therefore,
we can estimate the size of $H$, and if $|H|>(\eps t)^{2/3}$ we can safely output \badadvice.
Otherwise, if $|H|$ does not exceed the bound (so we do not reject), then it is small enough that --  even in high-arboricity graphs -- it cannot induce too many heavy triangles, thereby preserving Properties (i) and (iii). For Property (ii), once the set $R$ is sampled, we explicitly verify that $d(R)\eqdef \sum_{e\in R}d(e) =O(\ta |R|)$; if this fails, we reject the advice.
These modifications yield the runtime bound  \[O^*(m\tau_t/t+m\ta/t +\tau_d/\tau_t)=O^*(m/t^{2/3}+m\ta/t)\;.\]

\subsubsection{The $\Omega(m/t^{2/3})$ lower bound}

We prove a lower bound of $\Omega(m/t^{2/3}+m\aG/t)$ for any testable triangle-counting algorithm in the augmented model.
The $\Omega(m\aG/t)$  term follows from  the no-advice  lower bound of~\cite{ERS_faster}, and it remains to prove the second term, which is dominant whenever $\ta<t^{1/3}$.

Consider a testable algorithm invoked with advice $\ta < t^{1/3}$. The lower bound  is shown by constructing two contrasting families of graphs.

In the first family, every graph has arboricity $\ta$ and contains no triangles. The second family is obtained by taking graphs from the first family and augmenting them with a clique of size about $t^{1/3}$. (In the formal proof, the construction is refined to ensure that graphs in both families have exactly the same numbers of vertices and edges.) Consequently, graphs in the second family contain $\Theta(t)$ triangles and have arboricity $\Theta(t^{1/3})$, and so the advice is incorrect.

By definition, any testable algorithm must, with high probability, output an estimate $\hat{t} = 0$ on graphs from the first family, and either \badadvice{} or an estimate $\hat{t} \in (1 \pm \eps)t$ on graphs from the second family. Hence, the algorithm must be able to distinguish between the two families. This requires it to hit the planted clique with high probability, which in turn forces it to perform $\Omega(m/t^{2/3})$ samples.

\subsection{Additional related work}
\paragraph{Testable learning algorithms.}
Testable learning is a framework introduced
by Rubinfeld and Vasilyan
\cite{RubinfeldV23} to allow a user to safely use
a learning algorithm that is designed for specific
distributions on labelled examples (e.g., uniform).
The difficulty is that  ascertaining whether examples come
from the assumed distribution requires way too many samples.
The new framework suggests the design of a tester and learner
as a pair, where the tester efficiently tests a 
weaker property of the distribution, chosen such that the
weaker property is both easier to test
and sufficient for proving that the
algorithm's output is correct.
Their work, and several followups \cite{klivans2023testable,goel2024tolerant, slot2024testably, gollakota2024agnostically,gollakota2023tester,diakonikolas2023efficient,diakonikolas2024testable}
have demonstrated such tester-learner pairs for a number of
learning tasks. The work of \cite{goel2025testing} extends the testable learning framework to handle assumptions on the label noise in classification. 
A recent line of work on Testable Learning with Distribution Shift (TDS learning) \cite{klivans2023testable,klivans2024learning,chandrasekaran2024efficient, goel2024tolerant} tests the assumption that the data encountered by a classifier during deployment comes from the same distribution encountered during training.

In concurrent work, Marcussen, Rubinfeld and Sudan introduce an  analogue of 
testable  algorithms for the setting of random graphs
\cite{marcussenRS2025quality}, referred to as ``quality
control problems''.   
As in this work, the algorithms in \cite{marcussenRS2025quality}
have the goal of reliably approximating
the number of motifs in input graphs.
However, 
 the requirements of the testable algorithm differ in our setting and theirs, particularly with respect to the 
completeness guarantees of the algorithms. Completeness in their setting corresponds to algorithmic correctness 
with high probability over random graphs, while in our setting the algorithm must be correct with high probability on 
any low-arboricity input graph. Our work 
also uses very different techniques because
low-arboricity graphs have disparate
structural properties from random graphs.

\paragraph{Testing and estimating the arboricity in sublinear-time and space.}
The task of estimating the arboricity was studied in a series of works, 
both in the sublinear time and the streaming settings~\cite{bhattacharya2015space,mcgregor2015densest,king2023computing, bahmani2012densest, bhattacharya2015space}. Eden, Mossel and Ron~\cite{EstArb} improved the running times and space complexity to $O(n/\oa)$ at the cost of increasing the approximation parameter to $O(\log^2 n)$ and the  round complexity to logarithmic. 
 Eden, Levi, and Ron~\cite{testing_arboricity_ELR} presented an $O(1/\eps)^{O(\log(1/\eps))}$-time algorithm for testing whether a graph has bounded arboricity in the augmented model. However, being $\eps$-close to a bounded-arboricity graph is not sufficient for applying ERS: even a small number of additional edges can still generate many triangles.

\paragraph{Sublinear-time and space subgraph counting algorithms.}
Sublinear time subgraph counting results have been extensively studied both in various query models~\cite{GRS11,ELRS,ERS_cliques,hamiltonian,AKK,Peng,BiswasER21,Aliak,ERS19_dd,bishnu2025arboricity}.
Estimating the number of edges in sublinear time, 
when given access to uniform edge samples as well as degree and neighbor queries, follows (implicitly) from the work of Motwani, Panigrahy and Xu~\cite{motwani2007estimating} and from the work of Beretta and T{\v{e}}tek~\cite{beretta2024better}. 
Both these results give an essentially optimal bound of $\Theta^*(n^{1/3})$. 
When pair queries are also allowed, implying full augmented query access, the bound was recently improved to 
 $\Theta^*(n^{1/4})$ by Beretta, Chakrabarty and Seshadhri~\cite{beretta2025faster}.

Sublinear arboricity-dependent bounds were studied in~\cite{ERS19_dd,ERS_faster,assadi_lbs,bishnu2025arboricity,bera2020degeneracy,fichtenberger2022approximately}.


\paragraph{Triangle counting with predictions.}
Algorithms with predictions, like testable algorithms, use auxiliary information to overcome worst-case lower bounds. However, the auxiliary information is often more informative (as it models advice given by a machine learning algorithm).
In addition, such algorithms must always produce an output, even when the prediction is
of low quality.
The goal is that the performance improves with prediction accuracy, and, ideally,  never falls below the baseline guarantee.
Triangle counting in the streaming model using predictions was studied in several papers using various types of predictions~\cite{chen2022triangle,boldrin2024fast,boldrin2024fast}.
Common to all is that they require per-edge advice, rather than global information.

\subsection{Organization}

In Section~\ref{sec:preliminaries}, we formally define the notion of a testable algorithm, and the augmented query access model, and provide some useful notations and claims. 
In Section~\ref{sec:triangles}, we provide our sublinear  time and space algorithms and lower bound for triangle counting, with some of the proofs deferred to Appendix~\ref{sec: deferred proofs}. The sublinear-time lower bound is proven in Section~\ref{sec:triangles_lb}.
The streaming result is discussed in Section~\ref{sec: streaming}.
Finally, in Section~\ref{sec:edges}, we provide the algorithm for edge counting.

\section*{Acknowledgments}
We thank Sepehr Assadi for discussions that influenced this work.

\section{Preliminaries}\label{sec:preliminaries}
\subsection{Testable algorithm definition}

The following is our definition of testable algorithms for estimating the number
of motifs  with arboricity advice.  The problems of estimating the number of edges
and the number of triangles are a special case.

\begin{definition}[Testable algorithm    for motif counting  with arboricity advice]
\label{def:motif_counting_algorithm} 
Let $H$ be a fixed subgraph, and let $t$ denote the number of copies of $H$ in an input graph $G$.
A \emph{testable arboricity advice algorithm} for counting the number of copies of $H$ in $G$ (in the augmented query model)  is an algorithm that receives:
augmented query access to $G$, 
    an approximation parameter $\eps$, and 
 a positive integer $\ta$, which serves as a (possibly unreliable) upper bound on the arboricity of $G$.
 
The algorithm outputs either an estimate $\hat{t}$ or the comment \badadvice, as follows:
\begin{itemize}
    \item \textbf{Completeness:} If $\aG \leq \ta$, then whp $\calA$ outputs an estimate $\hat{t}$ (and does not output \badadvice).
    \item \textbf{Soundness:} For every graph $G$ and any $\ta$, if $\calA$ outputs a value $\hat{t}$, then whp
    \(
    \hat{t} \in (1 \pm \eps) t.
    \)
\end{itemize}
\end{definition}







\subsection{The augmented query model}

In the \emph{augmented query model}
the algorithm interacts with an unknown simple graph \(G=([n],E)\) through the following oracles:
  \textbf{Vertex access.} Vertices are labeled \([n]\), so the algorithm can specify any \(v\in[n]\) (equivalently, it can draw a uniform random vertex by choosing \(v\) uniformly from \([n]\)).
  \textbf{Degree query.} On input \(v\), return the degree of $v$, \(d(v)\).
  \textbf{\(i\th\)-neighbor query.} On input \((v,i)\), return the \(i\)th neighbor of \(v\) if \(1\le i\le d(v)\) (with respect to a fixed but arbitrary ordering of \(v\)'s incident edges), and a distinguished null symbol \(\bot\) otherwise.
  \textbf{Pair query.} On input \((u,v)\), return whether \(\{u,v\}\in E\).
  \textbf{Uniform edge sample.} Return an edge drawn uniformly at random from \(E\), independently at each call.

If no edge sample access is allowed, then this is referred to as the \emph{general access model.}

\subsection{Notations and useful inequalities}

\begin{definition}[Arboricity of a graph] The \emph{arboricity} of a graph $G$ is the minimum number of forests required to cover its edge set.
By~\cite{nash1961edge,nash1964decomposition}, it is also equivalent to
\[
\aG=\max_{S\subseteq G}\left\lceil \frac{m(S)}{n(S)-1} \right\rceil,
\]
where $n(S)$ and $m(S)$ denote the number of vertices and edges, respectively, in the subgraph induced by the set $S$.
\end{definition}

\begin{theorem}\label{thm:arb_inequalities_edges}
    For any graph $G$, $\aG\leq \sqrt m$ and $m\leq n\aG$.
\end{theorem}

The following theorem  of Chiba and Nishizeki bounds the number of triangles by the
arboricity.

\begin{theorem}[Due to \cite{ChNi85}]\label{thm:CN}
Let $G$ be a graph with arboricity at most $\aG$. Then 
$$t(G)\leq \sum_{e\in E} d(e)\leq 2m\cdot \aG\leq 2m^{3/2}.$$
\end{theorem}

\begin{definition}[Degree, neighbors and triangles of an edge]\label{def:edge_deg}
The degree of an edge $e=\{u,v\}$ is the \emph{minimum degree} of its endpoints, $d(e)=d(\{u,v\})=min\{d(u),d(v)\}$. The set of neighbors of an edge, is the set of vertices
that  are neighbors of its \emph{min degree endpoint}.

Also let $t(e)$ denote the number of triangles incident to the edge $e$.    
\end{definition}

Chernoff inequality and a useful corollary  are deferred to Appendix~\ref{sec:missing prliminaries}.

\section{The Triangles Approximation Algorithm}
\label{sec:triangles}
This section is devoted to proving the main theorem of the paper,
 a testable triangle counting algorithm with 
arboricity advice (Definition~\ref{def:motif_counting_algorithm}). We start with an algorithm that takes as an input a value $\ot$ that serves as a rough guess of the true triangle count $t$. In Section \ref{sec: Search Theorem for Triangle Counting} we show how to use Theorem \ref{thm: main theorem before search for ot} together with the Search Theorem of \cite{ERS_cliques} to search for a good guess $\ot$ and obtain a testable triangle counting algorithm that does not require an input of $\ot$.

Before presenting the algorithm, we provide some useful notations and structural claims. 

\subsection{Preliminaries for triangle counting}


    

The algorithm distinguishes those edges which either have high degree or too
many triangles incident to them:  
\begin{definition}[Heavy and light edges]\label{def:heavy}
For a given advice
$\ta$, and a guessed value $\ot$, let $\gamma=\max\{\ta, \ot^{1/3}\}$.
Further 
let $\tau_d=\settauD$ be a degree threshold, and $\tau_t=\settauT$ be a triangles-degree threshold.
Let $H_d$ denote the set of edges for which $d(e)> \tau_d$, and $H_t$ the set of edges for which $t(e)> \tau_t$.
We say that an edge $e$ is \emph{degree- (triangles-) heavy}, if either $e\in H_d$ ($e\in H_t$). Otherwise, we say it is \emph{degree (triangles) light}. Finally, we let $H=H_d\cup H_t$ and refer to it as the set of \emph{heavy edges}, and to $E\setminus H$ as the set of \emph{light edges}.
\end{definition}

The algorithm assigns each triangle to one of its edges as follows:
\begin{definition}[Assigning triangles to edges]\label{def:assign}
    Let $P=(E_0, E_1)$ be a partition of the graphs' edges. 
    We assign each triangle to its first edge in $E_0$ if such exists (according to some arbitrary predefined order). Otherwise, if all of its edges are in $E_1$, then the triangle is not assigned to any edge.
    For an edge $e$, we let $a_P(e)$ denote the number of triangles assigned to $e$
    by partition $P$.
\end{definition}

The algorithm partitions edges into two groups $E_0,E_1$. $E_0$ contains
all light edges. $E_1$ contains all edges that are degree-heavy  and all edges
that are  ``extremely" triangles-heavy -- that is, with a higher threshold than required
to just be triangles-heavy.  Edges that are degree-light, and triangles-heavy
with $\tau_t \leq t(e) \leq 2 \tau_t $ can be in either $E_0$ or $E_1$.
\begin{definition}[A good partition]\label{def:good_partiton}
We say that a partition of the graphs' edges $P=(E_0, E_1)$  is \emph{($\tau_d,\tau_t$)-good} if: (1) any light edge (according to Definition~\ref{def:heavy}) is in $E_0$,
(2)  every edge such that $d(e)> \tau_d$ or $t(e)> 2\tau_t$ is in $E_1$, and
(3) all other edges can be either in $E_0$ or $E_1$.
\end{definition}

Our first claim shows that if the advice is good, then we can ignore triangles for which all edges are heavy,
because there are not many of them.

\begin{claim} \label{clm:assigned-assuming-arbor}
    If $\ta>\aG$ and $\ot\in [t/4,t]$, then   $|H_d|\leq (\eps \ot)^{2/3}$ and $|H_t|\leq (\eps \ot)^{2/3}$. 
\end{claim}
\begin{proof} 
 By Theorem~\ref{thm:CN} and the assumption on $\ta$, $\sum_{e\in E}d(e)\leq 2m\cdot \aG\leq 2m\ta$. Therefore, there could be at most $2m\ta/\tau_d$ edges with degree greater than $\tau_d$.
 By the setting of $\tau_d=\settauD$, we get
 \begin{equation*}
    |H_d|\leq \frac{2m\ta}{\tau_d}=\frac{\eps\ta\ot}{4\gamma^2}\leq 
    \frac{\eps \ot}{4\gamma}\leq \frac{1}{4}\cdot\min\left\{(\eps \ot)^{2/3},\frac{\eps \ot}{\ta}\right\},
 \end{equation*}
  where the last two inequalities are since $\gamma=\max\{\ta, \ot^{1/3}\}$.

By the definition of $H_t$,
it holds that $\sum_{e\in H_t}t(e)\geq |H_t|\cdot \tau_t$ and also, $\sum_{e\in H_t}t(e)\leq 3t$ since this holds for any set of edges, so in particular it holds for $H_t$. 
Hence, by the setting of $\tau_t=\settauT=\frac{12}{\eps}\cdot \max\{\ta, \ot^{1/3}\}$ and since $t\leq 4\ot$,
it follows that
$
|H_t|\leq 3t/\tau_t\leq \eps 
 t/(4\ot^{1/3})\leq (\eps \ot)^{2/3}.
 $
 \end{proof}

\begin{claim} \label{clm:assigned}
   For any  partition $P=(E_0,E_1)$,
    it holds that $\sum_{e\in E}a_P(e)\leq t$.
    Furthermore, if
    $|E_1| \leq  3(\eps \ot)^{2/3}$ 
    and $\ot\in [t/4,t]$, then $\sum_{e\in E_0}a_P(e)\geq (1-12\eps)t$. 
\end{claim}
\begin{proof} 
By the assignment rule in Definition~\ref{def:assign}, every triangle is assigned to at most one edge, so it always holds that $\sum_{e\in E}a_P(e)\leq t$.

We now prove the second part of the claim. If there are at most $3(\eps \ot)^{2/3}$ edges in $E_1$, by Theorem~\ref{thm:CN} these edges induce at most $2\cdot |E_1|^{3/2}\leq 12\eps \ot\leq 12\eps t$ triangles, where the last inequality is by the assumption $\ot\leq t$. 
That is, there are at most $ 12\eps t$ triangles with all three edges are in $E_1$, 
and all the rest of the (at least) $(1-12\eps)t$ triangles will be assigned to one of the edges in $E_0$. 
Hence, $\sum_{e\in E_0} t(e)>(1-12\eps) t$. 
\end{proof}

\subsection{The algorithm}
\medskip
\alg{alg:triangles}{
   \textsf{Approx-Triangles-With-Arboricity-Advice}$(\ot, \eps, \delta,\ta)$
\begin{enumerate}
\item Set $\gamma=\max\{\ta, \ot^{1/3}\}$, $\tau_d=\settauD$ and $\tau_t=\settauT$.
\item \label{step: choose R}Sample a set $R$ of $r=\setq$ edges uniformly at random.
\item Query the degrees of all edges in $R$, and let $d(R)=\sum_{e\in R}d(e)$. \\If $d(R)>|R|\cdot\ta\cdot \frac{4}{\delta}$ then output $\badadvice$\label{step:R_high_deg}
\item Invoke \isHeavy{} on all edges in $R$, and if 
the fraction of  edges for which the algorithm returned ``Heavy" is greater than $\frac{5}{2}\cdot \frac{(\eps \ot)^{2/3}}{\om}$
then output \badadvice.
\label{step:reject-if-bad-edges}
\item[] \comment{Continue with the  ``standard'' algorithm by ERS.} 
\item Set up a data structure to allow sampling each edge in $R$ with probability $d(e)/d(R)$.
\item For $i=1$ to $s=\frac{d(R)}{|R|\cdot (\ot/m)}\cdot\frac{10\ln(8/\delta)}{\eps^2}$ do:
\label{step: main loop of main algo}
\begin{enumerate}
    \item Sample an edge $e\in R$, so that each edge is sampled  with probability $d(e)/d(R)$. \label{step:sample-edge-from-Q} 
    \item Pick $w$ to be a uniform random  neighbor of $e$.
    \item If $e,w$ form a triangle, then invoke \IsAssigned$(e,w)$. \label{step:invoke-is-Assigned} 
    \item If the triangle $(e,w)$ is assigned to $e$, then let $\chi_i=1$. Otherwise, let  $\chi_i=0$.
\end{enumerate}
\item Let  $\chi=\frac{1}{s}\sum_{i=1}^s\chi_i$
\item Return  $\hat{t}\eqdef \frac{d(R)\cdot m}{|R|}\chi$ 
\end{enumerate}
}
 \alg{alg:IsAssigned}{
   \noindent\textsf{IsAssigned}$(e=(u,v),w)$
\begin{enumerate}
\item Invoke \isHeavy{} for each of the edges $(u,v), (u,w), (v,w)$.
\item If $e$ is the first edge in $E_0$ then return YES. Otherwise, return NO.
\end{enumerate}
}

\alg{alg:IsHeavy}{
   \noindent\textsf{IsHeavy}$(e=(u,v), \tau_t, \tau_d)$
\begin{enumerate}
\item If $d(e)> \tau_d$ then return ``Heavy''.
\item Sample $k=18\frac{d(e)}{\tau_t} \ln\frac{10m}{\delta}$ neighbors of $e$ uniformly at random (from its min-degree vertex), and for each check if it forms a triangle with $e$. \\ \comment{If \textsf{IsHeavy$(e=(u,v), \tau_t, \tau_d)$} was called in the past, re-use the same random bits as used previously.}
\item If the number of witnessed triangles is greater than $1.5 k\cdot \frac{\tau_t}{d(e)}$, then return ``Heavy''. Otherwise, return ``Not Heavy''.
\end{enumerate}
}

\begin{theorem}
\label{thm: main theorem before search for ot}
Let $G$ be the input
graph  and $\eps$ and $\delta$ be an approximation and error parameters.
Let $t$ be the number of triangles in $G$. 
Algorithm \ApproxTriangles{} receives:
augmented query access to $G$, the number of edges $m$,
parameters $\eps$ and $\delta$, and 
a pair of positive integers $\ta$ and $\ot$.

The algorithm has an expected run-time and query complexity of $O\left(\frac{m\cdot\max\{\ta, \ot^{1/3}\}}{ \ot}\cdot \frac{\ln \frac{1}{\delta} }{\eps^3\delta} \ln\frac{m}{\delta}
 \cdot (1+t/\ot)
 \right)$. Every time the algorithm is run, it can either return a number $\hat{t}$
or it can return \badadvice, and
\begin{itemize}


\item  (Completeness)
For every $G$ and $
\ot$, if $\aG\leq  \ta$,
then w.p. at least $1-\delta$ the algorithm \ApproxTriangles{} will not output $\badadvice$.

\item  (Used for soundness)
For every $G$ and $
\ta$, if
$\ot \in [  t/4, t ]$ 
then  w.p. at least $1-\delta$ \ApproxTriangles{} will either return $\badadvice$ or output $\hat{t} \in (1 \pm 20\epsilon)t$.

\item  
(Used for soundness)
For every $G$ and $
\ta$, if
$\ot > t $ 
then w.p. at least $5\epsilon$, the algorithm \ApproxTriangles{} will either return $\badadvice$ or output a value of $\hat{t}$ such that $\hat{t} < (1+20\epsilon) t$.


\end{itemize}

\end{theorem}

The proof of the run-time bound in Theorem \ref{thm: main theorem before search for ot} is deferred to Appendix~\ref{sec:run_time_bound} and the correctness properties in Theorem are shown in Section \ref{sec: correctness of triangle counting} (this includes the completeness property, and the two properties used for soundness). Section \ref{sec: is Heavy} focuses on the sub-routine \textsf{IsHeavy} with the proof of Claim \ref{claim: isHeavy gives good partition} deferred to Appendix~\ref{sec: deferred proofs}.

As mentioned earlier, Theorem \ref{thm: main theorem before search for ot} is used in Section \ref{sec: Search Theorem for Triangle Counting} together with the Search Theorem of \cite{ERS_cliques} to obtain a testable algorithm (Definition \ref{def:motif_counting_algorithm}) in the augmented query model for $(1\pm \epsilon)$-approximate counting of triangles (Theorem \ref{thm: main theorem after search for ot}). The instance-adaptive algorithm (Corollary \ref{cor: triangle-counting algorithm that adapts to arboricity}) is analyzed in section \ref{sec: instance adaptie algo}.


 \subsection{Analyzing procedure \textsf{IsHeavy}}
 \label{sec: is Heavy}
 We rely on the following the guarantees of the procedure \isHeavy{}, and defer the proof  to Section~\ref{sec:proof-is-heavy}.
\begin{claim}
\label{claim: isHeavy gives good partition}
    With probability at least $1-\delta/10$ over all random coins, \isHeavy{} induces a $(\tau_d, \tau_t)$-good partition, as defined in Definition~\ref{def:good_partiton}.
    The query complexity and the run-time of \isHeavy{} is $O\left(\frac{\min(\tau_d,d(e))}{\tau_t} \ln\frac{m}{\delta}\right)$.
\end{claim}

\subsection{Correctness of Approx-Triangles-With-Arboricity-Advice}
\label{sec: correctness of triangle counting}
In this section we prove the Completeness condition and the two Soundness conditions in Theorem \ref{thm: main theorem before search for ot}. For the rest of the analysis we will use the following convention:
\begin{remark}
\label{remark: random variables are defined even if not computed}
    In this section we will consider as random variables various quantities in Algorithm \ApproxTriangles such as $\hat{t}$ and $\chi_i$. These random variables are well-defined even if the algorithm decides to terminate in steps  \ref{step:reject-if-bad-edges} or \ref{step:R_high_deg} prior to computing these quantities. This is true because   
    all these random variables are determined on (a) the partition defined by \isHeavy (b) the set $R$ and (c) the edges sampled in Step \ref{step: main loop of main algo}.    
\end{remark}
\subsubsection{Proof of Completeness}
In this sub-section we prove that the algorithm \TrianglesApprox{} satisfies the completeness condition in Theorem \ref{thm: main theorem before search for ot}. Specifically, we show that for every $G$ and $
\ot$, if $\aG\leq  \ta$,
then w.p. at least $1-\delta$ the algorithm \ApproxTriangles{ will not output $\badadvice$. Recalling that the algorithm has the option to reject in Step  \ref{step:R_high_deg} and Step \ref{step:reject-if-bad-edges}, in this section we show that either is unlikely if $\aG\leq  \ta$. We start by considering Step \ref{step:reject-if-bad-edges}.

\begin{claim}
\label{claim: estimating fraction of heavy edges}
   If the partition $P=(E_0,E_1)$ induced by the procedure \isHeavy{} satisfies $E_1\leq2(\eps \ot )^{2/3}$, then the algorithm outputs \badadvice{} in Step~\ref{step:reject-if-bad-edges} with probability only at most $\delta/4$. 
    Moreover, 
    $P=(E_0,E_1)$
    satisfies  
    $|E_1|\geq3(\eps \ot)^{2/3}$, then the algorithm will output \badadvice{} in Step~\ref{step:reject-if-bad-edges} with probability at least $1-\delta/4$.

\end{claim}
\begin{proof}
In Step~\ref{step:reject-if-bad-edges}, the procedure  \isHeavy{} outputs ``Heavy'' for en edge $e$ if and only if $e\in E_1$, where recall that $P=(E_0,E_1)$ induced by the procedure \isHeavy{}.
Therefore each edge $e$ sampled in Step~\ref{step:reject-if-bad-edges}, has a probability of $|E_1|/m$ to be in $E_1$ (so \isHeavy{} outputs ``Heavy'' for $e$). 

Suppose $E_1\leq2(\eps \ot)^{2/3}$, then
by the multiplicative Chernoff bound,  and  by the setting of $|R|\geq\setQ$,
 \[
 \Pr\left[
 \sum_{e\in R}\mathbbm{1}_{e\in E_1}
 >\frac{5}{2}\cdot |R|\cdot \frac{(\eps \ot)^{2/3}}{m}\right]<\exp\left(-\frac{|R|\cdot(\eps \ot)^{2/3}}{24m}\right)=\delta/4,
 \]
proving the first part of the claim.
 Now, suppose that $|E_1|\geq3(\eps \ot)^{2/3}$. 
 Then, again
 \[
 \Pr\left[
 \sum_{e\in R}\mathbbm{1}_{d(e)>\tau_d}
\leq \frac{5}{2}\cdot |R|\cdot \frac{(\eps \ot)^{2/3}}{m}\right]<\exp\left( 
 -\frac{1}{24}\cdot \frac{2(\eps \ot)^{2/3}}{m}\cdot |R|
\right)=\delta/4.
 \]
\end{proof}
Building on the claim above, we now show that Step \ref{step:reject-if-bad-edges} is unlikely to output $\badadvice$ if $\ta\geq \aG$.
\begin{claim}
    If $\ta\geq \aG$ and the partition $P=(E_0,E_1)$ induced by IsHeavy is a $(\tau_d, \tau_t)$-good partition (see Definition~\ref{def:good_partiton}), then the algorithm can output \badadvice{} in Step~\ref{step:reject-if-bad-edges} only with probability at most $\delta/4$. 
\end{claim}
\begin{proof}
If $P=(E_0, E_1)$ induced by IsHeavy is a $(\tau_d, \tau_t)$-good partition (Definition~\ref{def:good_partiton}), then the $E_1$ is a subset of $H_d\cup H_t$.
But if $\ta\geq \aG$, then
by Claim \ref{clm:assigned-assuming-arbor}, we see that $|H_d|\leq (\eps \ot)^{2/3}$ and $|H_t|\leq (\eps \ot)^{2/3}$, and therefore 
\[
|E_1|\leq |H_d|+|H_t|\leq 2 (\eps \ot)^{2/3},
\]
and thus by claim \ref{claim: estimating fraction of heavy edges} the algorithm outputs \badadvice{} in Step~\ref{step:reject-if-bad-edges} with probability only at most $\delta/4$. 
\end{proof}
\medskip

  Finally, we prove that if the advice $\ta$ is good, then with high probability the algorithm will not reject due to Step~\ref{step:R_high_deg}.
\begin{claim}
    If $\ta\geq \aG$, then the algorithm can only return \badadvice{} in Step~\ref{step:R_high_deg} with probability at most $\delta/3$.
\end{claim}
\begin{proof}
    By Theorem~\ref{thm:CN}, if the graph has arboricity at most $\aG$ and $\ta\geq \aG$, then
    \[
    \EX_{e\in E}[d(e)]=\frac{1}{m}\sum_{e\in E} d(e)\leq 2\aG\leq 2\ta.
    \]
    Hence, for every $i$, 
    $\Ex[d(R)]\leq |R|\cdot 2\ta$, and by Markov's inequality, 
    $\Pr[d(R)>4\ta|R|/\delta]<\delta/4$.
\end{proof}

\subsubsection{Proof of Soundness Condition 1}
In this sub-section we prove that the algorithm \TrianglesApprox{} satisfies the first soundness condition in Theorem \ref{thm: main theorem before search for ot}. Specifically, we prove that
for every $G$ and $
\ta$, if
$\ot \in [  t/4, t ]$ 
then, with probability at least $1-\delta$, \ApproxTriangles{ will either return $\badadvice$ or output $\hat{t} \in (1 \pm 20\epsilon)t$.

Let $P=(E_0,E_1)$ be the partition induced by the procedure \isHeavy{}. By Claim \ref{claim: isHeavy gives good partition}, with probability at least $1-\delta/10$,  the partition $P$ is a $(\tau_d,\tau_t)$-good (see Definition \ref{def:good_partiton}). Since the probability distribution of the partition $P$ is (as a random variable) independent from the choice of set $R$ and edges that \ApproxTriangles{ picks in Step \ref{step: main loop of main algo}. Thus, if we condition on the event that $P$ is a $(\tau_d,\tau_t)$-good partition, this does not change the distribution of $R$ and edges in Step \ref{step: main loop of main algo}. Therefore, for the rest of this subsection we condition on the event that  $P$ is a $(\tau_d,\tau_t)$-good partition.

Now, suppose that $|E_1|>3(\eps \ot)^{2/3}$. Claim \ref{claim: estimating fraction of heavy edges} tells us that in this case the algorithm \ApproxTriangles{ will with probability at least $1-\delta/4$ output  $\badadvice$. 

Thus, it remains to consider the case that  $|E_1|\leq 3(\eps \ot)^{2/3}$. We show in Claim \ref{claim: main claim soundness 1} (provided 
$\ot \in [  t/4, t ]$) in this case the random variable $\hat{t}$ will satisfy $\hat{t} \in (1 \pm 20\epsilon)t$. This will finish the proof because, as follows from Remark \ref{remark: random variables are defined even if not computed}, in the event that the \emph{random variable} $\hat{t}$ satisfies $\hat{t} \in (1 \pm 20\epsilon)t$ then the algorithm \ApproxTriangles{ either outputs such $\hat{t}$ or outputs $\badadvice$ at an earlier step.

We first need the following claim telling us that the number of triangles assigned to set $R$ will be approximately proportional to the total number triangles $t$ (with high probability):
\begin{claim}
\label{claim:triangles assigned to R are proportional to total triangles}
If $P=(E_0,E_1)$ is a $(\tau_d, \tau_t)$-good partition (Definition \ref{def:good_partiton}),  satisfies $|E_1|\leq 3(\eps \ot)^{2/3}$, and $\ot\in [t/4,t]$, then with probability at least $1-\frac{\delta}{4}$, we have $\frac{a_P(R)}{|R|} \in \left[(1- 13\eps)\cdot \frac{t}{m}, (1+\eps)\cdot \frac{t}{m}\right]$.
\end{claim}
\begin{proof}
    If $P$ is a $(\tau_d, \tau_t)$-good partition, then  for any edge $e\in E$, $a_P(e)\leq 2\tau_t$. This is true since if $t(e)>2\tau_t$, then $e\in E_1$, and so $a_P(e)=0$, and otherwise, if $e\in E_0$, then $a_P(e)\leq t(e)\leq 2\tau_t$.
    Furthermore, if $\ot\leq t$, 
    then  $|R|\geq\frac{16m\cdot \tau_t \cdot \ln(4/\delta)}{\eps^2 \cdot \ot}\geq \frac{m\cdot (2\tau_t)}{t}\cdot\frac{3\ln(4/\delta)}{\eps^2}$.
    Therefore, by Corollary~\ref{cor:chernoff}, with probability at least $1-\frac{\delta}{4}$,
    $$\frac{1}{r}\sum_{e\in R}a_P(e)\in (1\pm \eps)\EX_{e\in E}[a_P(e)]\in (1\pm \eps)\cdot \frac{\sum_{e\in E_0}a_P(e)}{m}.$$
Claim \ref{clm:assigned} tells us that if
    $|E_1| \leq  3(\eps \ot)^{2/3}$ 
    and $\ot\in [t/4,t]$, then $t\geq \sum_{e\in E_0}a_P(e)\geq (1-12\eps)t$, which combined with the bound above yields the conclusion of the claim.
\end{proof}
 
Overall, recall that the partition $P=(E_0,E_1)$ induced by the algorithm \textsf{IsHeavy} is a $(\tau_d, \tau_t)$-good partition with probability at least $1-\frac{\delta}{4}$ (Claim \ref{claim: isHeavy gives good partition}), this subsection has the premise that $\ot \in [  t/4, t ]$ and we are focusing on the case that $|E_1|>3(\eps \ot)^{2/3}$. With all this in mind, 
the following claim finishes the proof in this subsection, showing that the estimate $\hat{t}$ will indeed be a good approximation to the true triangle count $t$.
\begin{claim}
\label{claim: main claim soundness 1}
    Suppose $\eps\leq 1/20$ and the partition $P=(E_0,E_1)$ induced by the algorithm \textsf{IsHeavy} is a $(\tau_d, \tau_t)$-good partition (Definition \ref{def:good_partiton}), and also suppose $|E_1|\leq 3(\eps \ot)^{2/3}$ and $\ot\in [t/4,t]$,
    then with probability at least $1-\delta/4$ the algorithm outputs a value $\hatt$ such that $\hatt\in (1\pm 20\eps)t$.    
\end{claim}
\begin{proof}
    Observe that conditioned on an edge $e\in R$ being sampled in Step~\ref{step:sample-edge-from-Q}, $\chi_i=1$ iff the triangle $(e,w)$ is assigned to $e$, which happens with probability $\Pr[\chi_i=1\mid e]=\frac{a_P(e)}{d(e)}$. Therefore,
    
    \begin{equation*}
    \label{eq: expectation of chi-i}
    \EX[\chi_i]=\sum_{e\in R}\frac{d(e)}{d(R)}\cdot \frac{a_P(e)}{d(e)}=\frac{a_P(R)}{d(R)}.
    \end{equation*}
   Since Claim~\ref{claim:triangles assigned to R are proportional to total triangles}, $\frac{a_P(R)}{R}\in \left[(1-12\eps)\frac{t}{m}, (1+\eps)\frac{t}{m}\right]$, we have $\EX[\chi_i]\in \left[(1-12\eps)\frac{|R|\cdot t}{d(R)\cdot m}, (1+\epsilon)\frac{|R|\cdot t}{d(R)\cdot m}\right]$.
    This, together with the premise that $\ot\in [t/4,t]$ and $\eps\leq 1/20$ implies that
    \[s\eqdef\frac{d(R)}{|R|\cdot \frac{\ot}{m}}\cdot\frac{10\ln(8/\delta)}{\eps^2}\geq \frac{1}{\EX[\chi_i]}\cdot \frac{3\ln(8/\odelta)}{\eps^2},\] and therefore the Chernoff bound (Corollary \ref{cor:chernoff}) allows us to conclude that
    \[
    \Pr\left[\left\lvert\frac{1}{s}\cdot \chi-\EX[\chi_i]\right\rvert>\eps\EX[\chi_i]\right]<\frac{\delta}{4}.
    \]
If this holds, then we have
\[
\hat{t}\eqdef
\frac{d(R)\cdot m}{|R|}\chi
=
(1\pm 13\epsilon)t,
\]
proving the claim.
\end{proof}
\subsubsection{Proof of Soundness Condition 2.}
In this sub-section we prove that the algorithm \TrianglesApprox{} satisfies the second soundness condition in Theorem \ref{thm: main theorem before search for ot}. This is accomplished by the following claim:
\begin{claim}
For every $G$ and $
\ta$, if
$\ot > t $ 
then w.p. at least $\epsilon/4$, the algorithm \ApproxTriangles will either return $\badadvice$ or output a value of $\hat{t}$ such that $\hat{t} < (1+20\epsilon) t$.
\end{claim}
\begin{proof}
   Let $P=(E_0,E_1)$ again denote the partition induced by the algorithm IsHeavy. Conditioned on a specific choice of $R$, the expectation of each $\chi_i$ still satisfies Equation \ref{eq: expectation of chi-i}, and therefore averaging over $R$ we have
   \[
   \EX[\hat{t}]=
    \EX\left[\frac{d(R)\cdot m}{|R|}\chi\right]
    =\EX\left[\frac{d(R)\cdot m}{|R|}\cdot \frac{a_P(R)}{d(R)}\right]
    =\frac{m}{|R|}\cdot\EX\left[a_P(R)\right].
   \]
   By the assignment rule in Definition~\ref{def:assign},  and by the first part of Claim~\ref{clm:assigned}, we have\\
    $\EX_{e\sim E}[a_P(e)]=\frac{1}{m}\sum_{e\in E} a_P(e)\leq \frac{t}{m}$, 
    and, by linearity of expectation, $\EX[a_P(R)]\leq |R| \cdot \frac{t}{m}$. Combining with the inequality above, we get $EX[\hat{t}]\leq t$.    
Therefore, by Markov's inequality, 
    we have
    $$\Pr[\hat{t}>(1+20\eps)t]< \frac{1}{1+\eps}<1-5\eps.$$
    Hence, with probability at least $5\eps$, the random variable $\hat{t}$ satisfies $\hat{t}\leq(1+20\eps)t$. Conditioned on this event, either the algorithm outputs such value of $\hat{t}$ satisfying $\hat{t}\leq(1+\eps)t$, or it terminates and returns  $\badadvice$ at an earlier step.  
\end{proof}
\subsection{From testable triangle-counting to   instance-adaptive counting}
\label{sec: instance adaptie algo}
In this section, we use Theorem \ref{thm: main theorem after search for ot} to obtain an instance-adaptive algorithm for approximate triangle counting (Corollary \ref{cor: triangle-counting algorithm that adapts to arboricity}).
    \begin{proof}[Proof of Corollary~\ref{cor: triangle-counting algorithm that adapts to arboricity}]
    We run the algorithm in Theorem \ref{thm: main theorem after search for ot} with a parameter $\odelta=\delta/(10 \log m)$ and an increasing sequence of values $\ta=2^j$ where $j\in\{1,2,3,\cdots, \log m\}$. Once for some $j$ the algorithm outputs value $\hat{t}_j$ (Rather than $\badadvice$), we output  $\hat{t}_j$ and terminate.
    

    The soundness condition in (Definition \ref{def:motif_counting_algorithm}), together with a union bound over $j$, implies that once we receive an estimate $\hat{t}_j$ at one of the steps $j$, this estimate will satisfy $\hat{t}_j\in (1\pm \eps) t$.

The true arboricity $\aG$ is at most $2\sqrt{m}$. For one of the iterations $j_0$, we have $\ta_{j_0}=2^{j_0} \in [\aG, 2\aG]$.
The completeness condition tells us that for each $j\geq j_0$ 
with probability at least $1-\delta$ the algorithm from Theorem \ref{thm: main theorem after search for ot} terminates. Thus, the expected run-time of the algorithm is upper-bounded by
\begin{multline*}
j_0 \cdot O^*\left(
\frac{m\cdot \ta \log 1/\delta }{t}
+
\frac{m \log 1/\delta}{t^{2/3}}
 \right)
 +
 \sum_{j>j_0}
 \delta^{j-j_0}
 O^*\left(
\frac{m\cdot \ta \cdot 2^{j-j_0} \log 1/\delta }{t}
+
\frac{m \log 1/\delta}{t^{2/3}}
 \right)\leq \\
 j_0 \cdot O^*\left(
\frac{m\cdot \ta \log 1/\delta }{t}
+
\frac{m \log 1/\delta}{t^{2/3}}
 \right)
 +
 \sum_{j>j_0}
 0.1^{j-j_0}
 O^*\left(
\frac{m\cdot \ta \cdot 2^{j-j_0} \log 1/\delta }{t}
+
\frac{m \log 1/\delta}{t^{2/3}}
 \right)\leq\\
  O^*\left(
\frac{m\cdot \ta \cdot 2^{j-j_0} \log 1/\delta }{t}
+
\frac{m \log 1/\delta}{t^{2/3}}
 \right),
\end{multline*}
finishing the proof.
    \end{proof}
\section{A Lower Bound for Testable Triangle Counting Algorithms}\label{sec:triangles_lb}

\begin{proof}[Proof of Theorem~\ref{thm:lb_triangle_counting}]
    Note that if $\ta<m/n$, then by Theorem~\ref{thm:arb_inequalities_edges}, any testable algorithm can simply reject, so that proving any non-trivial lower bound is impossible. Hence,  the lower bound only holds for the case where $\ta\geq m/n$.

    Since any testable algorithm for triangle counting also implies a triangle-counting algorithm in the case where the graph is guaranteed to be low-arboricity, the lower bound of $\Omega^*(\frac{m\cdot \aG}{t})$ by~\cite{ERS_faster,bishnu2025arboricity} applies to testable algorithms, and so it remains to prove  a lower bound  of $\Omega(m/t^{2/3})$ for the regime $\ta<t^{1/3}$. 
    
    Consider the following two families of graphs.
    In the first family, $\mG_1$, all graphs have $\Theta(n)$ nodes, $\Theta(m)$ edges, arboricity $\ta$ and zero triangles. In the second family $\mG_2$, all graphs have $\Theta(n)$ nodes, $\Theta(m)$ edges, arboricity $t^{1/3}$ and $\Theta(t)$ triangles.

    Observe that any testable algorithm must (whp) return $\hatt=0$ for graphs from the first family, and (whp) either output  \badadvice{} or a correct estimate of $t$.
    In particular, this implies that the algorithm must be able to distinguish between the two families with high probability.

    More concretely, in the first family the graphs all consist of two disjoint sets. The first, $A_1$, is a bipartite graph over $n$ vertices where each vertex has degree $\ta$, and the second, $A_2$ is a  bipartite set over $r=\Theta(t^{2/3}/\ta)$ vertices so that each vertex has degree $\ta$ and in total there are $\binom{t^{1/3}}{2}$ edges. 

     In the second family the graphs are very similar, but the set $A_2$ is a clique over $t^{1/3}$ vertices, so that here too it contributes $t^{2/3}$ edges, and also $\Theta(t)$ triangles.  In addition, the graphs in this family have an additional independent set, $A_3$ over $r-t^{1/3}$ vertices, so that both graph families have the same number of vertices and edges.
    Within each family the graphs only differ on the labelings of the vertices.

     In order to distinguish the two families with high probability, any testable algorithm must hit at least one vertex or edge from $V\setminus A_1$. Hence, it must make \[\Omega\left(\min\left\{\frac{n}{r},\frac{m}{t^{2/3}}\right\}\right)=\Omega\left(\min\left\{\frac{n\ta}{t^{2/3}},\frac{m}{t^{2/3}}\right\}\right)=\Omega\left(\frac{m}{t^{2/3}}\right)\]
     many uniform vertex or edge samples.
\end{proof}

\section{Testable Streaming Algorithm for Counting Triangles}
\label{sec: streaming}

In this section we prove that our testable estimation algorithm can be adapted to the streaming model, achieving $9$ rounds and space complexity
$O^*\left(\frac{m}{t^{2/3}}+\frac{m\ta}{t}\right)$.
The state-of-the-art arboricity-dependent bound is due to Bera and Seshadhri~\cite{bera2020degeneracy}, which uses only $6$ rounds and $O^*(m\oa/t)$ space, but requires a trusted bound on $\aG$.

As shown in~\cite{EstArb,fichtenberger2022approximately}, any sublinear-time algorithm in the augmented query model with $O(q)$ queries and adaptivity depth $k$ can be simulated in the arbitrary-order insert-only edge arrival streaming model using $O(q)\cdot \poly \log n$ space and $k$ \emph{rounds}.\footnote{An algorithm $\mA$ in the augmented query model has adaptivity depth $k$ if the following holds. For every execution of $\mA$, the set of queries $Q$ performed by it can be partitioned into $k$ sets $Q_1,\ldots,Q_k$, such that for each $j\in[k]$, the queries in $Q_j$ can be determined solely from the responses to $Q_1,\ldots,Q_{j-1}$.} That is, query complexity translates into space complexity, while adaptivity depth translates into the number of rounds. This is true since neighbor and edge queries can be implemented with $\ell_0$-samplers, degree queries with counters, and pair queries with a one-bit verifier. Hence, a set of non-adaptive queries can be simulated in a single pass. Inspecting the algorithm \ApproxTriangles{}, one can verify that it has an adaptivity depth of $9$.

Recall that \ApproxTriangles{} receives both an advice value $\ta$ for $\aG$ and a guessed value $\ot$ for the number of triangles in the graph. In the sublinear setting, the latter can be removed using a search procedure. However, this search cannot be implemented in the streaming model with a constant number of passes; it would either introduce a polylogarithmic round overhead or, if executed in parallel, require linear space. This limitation is common in streaming counting algorithms. Consequently, the streaming literature typically assumes access to a coarse constant-factor estimate of the triangle count (see, e.g.,~\cite{braverman2013hard,mcgregor2016better} for a detailed discussion on this assumption). Given such an estimate, \ApproxTriangles{} yields a constant-round streaming testable algorithm for approximating the number of triangles, as stated in Corollary~\ref{cor: streaming}.

A similar polylogarithmic round overhead would also arise if the algorithm were not provided with the advice value $\ta$ in advance, and instead had to adapt its space complexity to the true arboricity $\aG$ of the graph (cf.\ Corollary~\ref{cor: triangle-counting algorithm that adapts to arboricity})

\section{Approximating the Number of
Edges}\label{sec:edges}

In this section we prove Theorem~\ref{thm: testable edge estimation theorem} and its following corollary.

\begin{corollary}
\label{cor:edges}
    There exists an algorithm with the following guarantees. The algorithm receives access to a graph $G$, and is given the number of vertices $n$ in $G$ together with parameters $\epsilon$ and $\delta\in(0,1)$. With probability at least $1-\delta$, the algorithm produces an estimate $\hat{m}$ that with probability at least $1-\delta$ satisfies $\hat{m}\in (1\pm \epsilon)m$, where $m$ is the number of edges in $G$. The expected run-time of the algorithm is $O^*\left(
\frac{n\cdot \alpha(G) \log 1/\delta }{t}
 \right)$, where $\alpha(G)$ is the arboricity of $G$.    
\end{corollary}




Recall that in this section too the algorithm relies on access to 
uniform vertex samples, degree queries and neighbor queries, as well as 
uniform edge samples. However,  we remove the assumption on having prior knowledge of $m$, as this is the task at hand. Instead, the algorithm is given a guess $\om$, and we use the search theorem in order to output a good estimate of $m$ (see Section~\ref{sec:search_edges}).


We now present the algorithm which is a simple adaptation of~\cite{Seshadhri_simple_edges}, and the main theorem of this section.

\alg{alg:approx-edges-advice}{
    \noindent\textsf{Approx-Edges-With-Arboricity-Advice}$(\om,\eps,\oa,\delta)$
    \begin{enumerate}
    \item Let $\delta=\delta/2$ and $\eps=\eps/6$.
        \item Sample a set $R$ of $r=\frac{12\ln(1/\delta)}{\eps^2}$ edges uniformly at random.
        \item For every $e\in R$, let $x_e=1$ if $d(e)>2\oa/\eps$.
        \item If $\frac{1}{r}\sum_{e\in R}x_e>2\eps$ then output \badadvice. \label{step:reject-many-high-deg}
        \item (Otherwise, continue with the bounded-arboricity algorithm of \cite{Seshadhri_simple_edges})    
        \item Let $q=\frac{n\oa}{\om}\cdot \frac{12\ln(2/\delta)}{\eps^3}$.
    \item If $q\geq n$, query $d(v)$ for all $v\in V$ and \textbf{return} $\sum_{v\in V}d(v).$
    \item For $i=1$ to $q$:
    \begin{enumerate}
         \item Sample $u\in V$ uniformly at random.\label{step:samp_v_with_advice}
            \item Sample $v\in \Gamma(u)$ uniformly at random. \label{step:samp_nbr_with_advice}
                \item If $d(u)\leq 2\alpha/\eps$ and $u \prec v$, then let $\chi_{i}=d(u)$. Otherwise, let $\chi_{i}=0$.
            \end{enumerate}
            \item \label{step:edges_return_mhat} \textbf{Return} $\hm=\frac{n}{q}\sum_{i}\chi_i$.
    \end{enumerate}
}

\begin{theorem}\label{thm:edges-with-advice}
Let $G$ be the input
graph, and $\eps$ and $\delta$ the approximation and failure parameters.
Let $m$ be the number of  edges in $G$. In addition to 
augmented query access to $G$ 
the algorithm takes as an input a pair of positive integers $\ta$ and $\om$.
Then, 
the following holds.
\begin{itemize}

\item  
For every $G$ and $
\om$, if $\aG\leq  \ta$,
then w.p. at least $1-\delta$ the algorithm \textsf{Approx-Edges-With-Arboricity-Advice} will not output \badadvice.

\item  (used for soundness)
For every $G$ and $
\ta$, if
$\om <m$ 
then with probability at least $1-\delta$ \textsf{Approx-Edges-With-Arboricity-Advice} will either return \badadvice or output $\hat{m} \in (1 \pm 10\epsilon)m$.

\item  
(used for soundness)
For every $G$ and $
\ta$, if
$\om > m $ 
then w.p. at least $\epsilon/4$, the algorithm \approxEdges\ will either return \badadvice or output a value of $\hat{m}$ such that $\hat{m} < (1+\epsilon) m$.

\item The query complexity and running time are the minimum between $O\left(\frac{n\oa}{m}\right)\cdot\poly(1/\eps, \ln(1/\delta))$ and $O(n)$.
\end{itemize}

\end{theorem}

We start with some preliminary definitions and claims.

\subsection{Preliminary notations and claims}
\begin{definition}[Heavy vertices and edges.]
    Let $H$ denote the set of all vertices with $d(v)>2\oa/\eps$. Let $m(H)$ denote the set of edges in the subgraph induced by $H$. I.e., edges with both endpoints having degree greater than $2\oa/\eps$ (so that $d(e)>2\oa/\eps$).
    Let $L=V\setminus H$.
\end{definition}

\begin{notation}[An ordering and outgoing edges]
We assume an ordering on the graphs' vertices such that $u\prec v$ if $d(u)<d(v)$ of if $d(u)=d(v)$ and $id(u)<id(v)$.
    We let $d^+(u)$ denote the number of neighbors of $v$ such that $u\prec v$ and refer to this set as the set of \emph{outgoing neighbors}  of $u$. For a set $S$ we let $d^+(S)=\sum_{v\in S}d^+(S)$.

    Observe that 
        \[
    \sum_{v\in V} d^+(v)=d^+(L)+d^+(H)= m\;.
    \]
\end{notation}
    
\begin{observation}
Observe that since every $v\in H$ contributes at least $2\alpha/\eps$ edges to $m$, it holds that $|H|\leq \epsilon m/(2\alpha)$. Thus, by the Nash-Williams theorem, $m(H)\leq 2|H|\cdot \alpha\leq \epsilon m$. Therefore, 
\[m(H)=\sum_{v\in H}d^+(v)<\eps\cdot  m, \text{ so that } d^+(L)\eqdef \sum_{v\in L}d^+(v)\in [(1-\eps)m,m]\]
\end{observation}\label{obs:bound-e(H)}

\subsection{Proof of Theorem~\ref{thm:edges-with-advice}}

\begin{proof}[Proof of Theorem~\ref{thm:edges-with-advice}]
We prove each of the items separately.

\textbf{Proof of item 1.}
By Observation~\ref{obs:bound-e(H)}, if $\oa>\alpha(G)$, then $m(H)<\eps m$.
    Hence, for every $e\in E$, $\Pr_{e\sim E}[d(e)>2\oa/\eps]<\eps$.
    By Chernoff's inequality~\ref{thm:chernoff}, for $r=12\ln(1/\delta)/\eps^2$,
    \[
    \Pr\left[\frac{1}{r}\sum_{e\in R}x_e>2\eps\right]<\delta.
    \]
    Hence, the probability of rejecting in Step~\ref{step:reject-many-high-deg} is at most $\delta$.
\end{proof}

\textbf{Proof of item 2.}
If $m(H)>4\eps m$, then by Chernoff's inequality~\ref{thm:chernoff}, \[
    \Pr\left[\frac{1}{r}\sum_{e\in R}x_e<2\eps\right]<\delta.
    \]
Therefore, with probability at least $1-\delta$, the   algorithm returns ``Bad Advice" in Step~\ref{step:reject-many-high-deg} and the claim holds. 

Now consider the case where $m(H)<4\eps m$. Then 
$d^+(L) = m-d^+(H)\in [(1-4\eps),m]$.

    Let $u$ denote the vertex sampled in Step~\ref{step:samp_v_with_advice}.
    If  $u\in L$ then  $\EX_{v\sim\Gamma(u)}[\chi_{i}\mid u,\; u\in L\;]=\frac{d^+(u)}{d(u)}\cdot d(u)=d^+(u).$
    Also, for any $u\in H$, $\EX[\chi_i \mid u,\; u\in H\;]=0$.
    Hence, by linearity of expectation,
    \begin{equation}
    \label{eq:exp_chi_edges}    
    \EX_{u\in V, v\in \Gamma(u)}[\chi_i]=\frac{1}{n}\left(\sum_{u\in L} d^+(u)+\sum_{u\in H} 0\right)=\frac{d^+(L)}{n}.
    \end{equation}
    Furthermore, for any $u$, $\chi_i\leq 2\oa/\eps$. Therefore, 
    the $\chi_i$'s are $[0,B]$ random variables with $B=2\oa/\eps$ and $\mu =\EX[\chi_i]=\frac{d^+(L)}{n}\geq \frac{\frac12m}{n}$.
     Since $r=\frac{n\cdot (2\oa/\eps)}{\frac{1}{2}m}\cdot \frac{3\ln(2/\eps)}{\eps^2}\geq\frac{B}{\mu}\cdot \frac{3\ln(2/\eps)}{\eps^2}$, by Corollary~\ref{cor:chernoff},
     $\hat{\mu}\eqdef \frac{1}{r}\sum_{i=1}^r\chi_i \in (1\pm\eps)\cdot \frac{d^+(L)}{n}\in (1\pm 10\eps)\frac{m}{n}$. Hence, the returned value $\hat{m}=\frac{n}{r}\sum_{i=1}^r\chi_i$ in Step~\ref{step:edges_return_mhat}
    is with probability at least $1-\delta$ in $(1\pm10\eps)m$.
    
    \textbf{Proof of item 3.}
     By Equation~\ref{eq:exp_chi_edges} linearity of expectation,
    \[
    \EX\left[\frac{1}{q}\sum_{i}\chi_i\right]=\EX[\chi_1]=\frac{d^+(L)}{n}\leq \frac{m}{n}.
    \]
    By Markov's inequality, 
    \[
    \Pr\left[\frac{1}{q}\sum_{i}\chi_i > (1+\eps)m\right]<\frac{1}{1+\eps} < 1-\eps/4.
    \]

\textbf{Proof of item 4.}
The query complexity and running time are 

\[
r+\min{n,q}=O(\min\{n,q\})=\min\left\{n,\frac{n\oa}{m}\cdot \frac{\ln(1/\delta)}{\eps^3}\right\}\;.
\]

\subsection{Searching for the ``right" $\om$ value}\label{sec:search_edges}

Searching for the right $\om$ value can be done almost identically to the search for the right $\ot$ value in Section~\ref{sec: Search Theorem for Triangle Counting}. Here too we rely on Theorem~\ref{thm:search}, where here algorithm $\mA$ will be defined to be 
\[\mA(\om, \eps, \delta, \oa)=\textsf{\approxEdges}(\om, \eps/10,\delta, \oa)\;,\]
and redefine algorithm $\mathcal{B}(\eps,\oa )$ as follows:
\begin{itemize}
    \item 
    \search$\left(\mA,n^{2},\eps,\oa \right)$ with the choice of $\invokeA$ defined above.
    \item If at any point while running \search$\left(\mA,n^2,\eps,\oa \right)$ a call to $\invokeA$ returns ``Bad Advice",  output ``Bad Advice" and terminate. 
\end{itemize}

The rest of the proof is identical to the triangles case in Section \ref{sec: Search Theorem for Triangle Counting}.

\bibliographystyle{plain}
\bibliography{refs}

\appendix

\section{Missing Preliminaries}\label{sec:missing prliminaries}
We shall make use of Chernoff's inequality and its corollary.

\begin{theorem}[Chernoff inequality~\cite{chernoff}]\label{thm:chernoff}
Let \(\chi_1,\ldots,\chi_r\) be independent random variables in \([0,B]\) with
\(\EX[\chi_i]=\mu\).
For any \(\eps\in(0,1],\)
\[
  \Pr\left[\tfrac1r\sum_{i=1}^r \chi_i>(1+\eps)\mu\right]
  < e^{-\eps^{2}\mu r/(3B)},
  \qquad
  \Pr\left[\tfrac1r\sum_{i=1}^r \chi_i<(1-\eps)\mu\right]
  < e^{-\eps^{2}\mu r/(2B)}.
\]
\end{theorem}

\begin{corollary}\label{cor:chernoff}
    Let $\delta\in(0,1)$ and let  \(\chi_1,\ldots,\chi_r\) be independent random variables in \([0,B]\) with
\(\EX[\chi_i]=\mu\). Let $\hat \mu=\tfrac1r\sum_{i=1}^r \chi_i$.
If $r\geq \frac{B}{\mu}\cdot \frac{3\ln(2/\odelta)}{\eps^2}$, then with probability at least $1-\delta$, $\hat\mu\in (1\pm \eps)\mu$.

If $r<\frac{B}{\mu}\cdot \frac{3\ln(2/\odelta)}{\eps^2}$ so that $r=\frac{B}{\overline\mu}\cdot \frac{3\ln(2/\delta)}{\eps^2}$ for some value $\overline \mu >\mu$, 
then with probability at least $1-\delta$, $\hat\mu<(1+ \eps)\overline\mu$.
\end{corollary}

\section{Deferred Proofs For the Triangle Counting Algorithm}
\label{sec: deferred proofs}

\subsection{Proof of IsHeavy guarantees}\label{sec:proof-is-heavy}


\begin{proof}[Proof of Claim \ref{claim: isHeavy gives good partition}]
First of all, step (1) will return Heavy if and only if $e$ is a degree-heavy edge. 

We will now argue that for every specific non-degree-heavy edge $e$, if $t(e)\geq 2\tau_t$ then \isHeavy{} returns ``Heavy" with probability at least $1-\frac{\delta}{10m}$,  and if $t(e)< \tau_t$ then it will return ``Not Heavy" with probability at least $1-\frac{\delta}{10m}$. After this is proven, a union bound over all such edges tells us that with probability at least $1-\delta/10$ over the random coins the algorithm isHeavy will induce a $(\tau_d, \tau_t)$-good partition, as defined in Definition~\ref{def:good_partiton}.

Fix a specific non-degree-heavy edge $e$. If $d(e)\leq \tau_t$, then $1.5 k\cdot \frac{\tau_t}{d(e)}\geq k$ and therefore the algorithm will return ``Not Heavy". Since for this edge $t(e)\leq d(e)\leq \tau_t$ this is the behavior required by the definition of a good partition (Definition \ref{def:good_partiton}).

For the rest of the proof, suppose $d(e) < \tau_t$.
Then, each of the $k$ neighbors of $e$ has a probability of $t(e)/d(e)$ to be a part of a triangle. Let $Y$ denote the number of witnessed triangles. 
Set $\mu=\tau_t/d(e)$, and let $Z$ be a sum of $k$ i.i.d. Bernoulli variables in $\{0,1\}$ each of which has an expectation of $1.5 \mu$.
Chernoff's inequality tells us that for any $\gamma\in (0,1)$ we have
\[
  \Pr\left[\tfrac Zk\geq 2\mu\right]
  \leq  e^{-(1/3)^{2}\cdot 1.5\mu k/3}
  =
  e^{-\mu k/18}
  \leq \frac{\delta}{10m}
  ,
  \qquad
  \Pr\left[\tfrac Zk \leq \mu\right]
  \leq e^{-(1/2)^{2}\cdot1.5\mu k/2} = e^{-3\mu k/16}\leq \frac{\delta}{10m}.
\]
Hence, we can finish the proof by considering the following two cases: 
\begin{itemize}
    \item Suppose $t(e)< \tau_t$.
    Each random neighbor of $e$ has a probability of forming a triangle equal to $t(e)/d(e)\leq \tau_t/d(e)=\mu$. This implies that for every threshold $w$ we have\footnote{This is a consequence of a general observation that if $Y_1$ is a sum of $k$ Bernoulli random variables with expectation $\mu_1$ and $Y_2$ is a sum of $k$ Bernoulli random variables with expectation $\mu_2\geq \mu_1$ then for every threshold $w$ we have $\Pr[Y_1\geq w]\leq \Pr[Y_2\geq w]$. This observation follows immediately by replacing the $k$ Bernoulli random variables one at a time and observing that the probability of exceeding the threshold $w$ will not decease at every step.    
    } $\Pr[Y\geq w]\leq \Pr[Z\geq w]$, and in particular 
    \[
     \Pr\left[\tfrac Yk\geq \frac{2 \tau_t}{d(e)}\right]
     \leq 
     \Pr\left[\tfrac Zk\geq \frac{2 \tau_t}{d(e)}\right]
     =\Pr\left[\tfrac Zk\geq 2 \mu\right]
     \leq \frac{\delta}{10m}.
    \]
   Since the algorithm outputs ``Heavy" only when $\frac{Y}{k}>1.5 \frac{2 \tau_t}{d(e)}=1.5\mu$,
   the inequality above upper-bounds the probability of this event by at most $\frac{\delta}{10m}$.  
    \item Suppose $t(e)\geq 2\tau_t$.
    Again, each random neighbor of $e$ has a probability of forming a triangle equal to $t(e)/d(e)\geq 2\tau_t/d(e)=2\mu$. This implies that for every threshold $w$ we have $\Pr[Y\leq w]\leq \Pr[Z\leq w]$, and in particular 
    \[
     \Pr\left[\tfrac Yk\leq \frac{ \tau_t}{d(e)}\right]
     \leq 
     \Pr\left[\tfrac Zk\leq \frac{ \tau_t}{d(e)}\right]
     =\Pr\left[\tfrac Zk\leq  \mu\right]
     \leq \frac{\delta}{10m}.
    \]
    Since the algorithm outputs ``Heavy" only when $\frac{Y}{k}>1.5 \frac{2 \tau_t}{d(e)}=1.5\mu$, the inequality above lower-bounds the probability of this event by $1-\frac{\delta}{10m}$.  
\end{itemize}
This concludes the proof.
\end{proof}

\subsection{Run-time bound}\label{sec:run_time_bound}
In this sub-section we prove that the algorithm \TrianglesApprox{} satisfies the run-time bound in Theorem \ref{thm: main theorem before search for ot}.

To bound the expected run-time of the algorithm, we first bound the expected number of times that IsAssigned is invoked.
\begin{claim}
\label{claim: how many times IsAssigned invoked}
    In expectation, IsAssigned is invoked at most $O\left(\frac{\ln(1/\delta)}{ \eps^2}\cdot t/\ot\right)$ times.
\end{claim}
\begin{proof}
    For a specific value of the set $R$ chosen in Step~\ref{step: choose R},  and let $F_{i,R}$ denote the event that a triangle is detected in  the $i\th$ invocation of Step~\ref{step:invoke-is-Assigned}. 
    Let $e~\sim_d R$ denote the distribution induced on edges when sampling an edge in $R$, each with probability $\frac{d(e)}{d(R)}$.
    Observe that conditioned on $R$ and under this distribution,
    \[
    \EX_{e\sim_{d} R}[F_{i,R}\mid R]=\sum_{e\in Q}\frac{d(e)}{d(R)}\cdot\frac{t(e)}{d(e)}=\frac{t(R)}{d(R)} \,
    \]
    where $t(e)$ denotes the number of triangles incident to the edge $e$, and $t(Q)=\sum_{e\in E}t(e)$.
    Since $\EX_{e\sim E}[t(e)]=\frac{3t}{m}$, 
    it holds by  linearity of expectation that 
    $\EX[t(R)]=|R| \cdot \frac{3t}{m}$.
    Hence, removing the conditioning on $R$, 
    $\EX_{R,e\sim_{d} R}[F_{i,R}]=\frac{|R|\cdot \frac{3t}{m}}{d(R)}$.
    Hence, \[\EX\left[\sum_{i=1}^{s} F_{i,R}\right]=
    s\cdot\frac{|R|\cdot \frac{3t}{m}}{d(R)}=
    \frac{d(R)}{|R|\cdot \frac{\ot}{m}}\cdot \frac{10\ln(1/\delta)}{\eps^2}\cdot \frac{|R|\cdot \frac{3t}{m}}{d(R)}=O\left(\frac{\ln(1/\delta)}{\eps^2}
    \cdot (t/\ot)
    \right).\]

    Hence, the expected number of invocations of \textsf{IsAssigned} is $O\left(\frac{\ln(1/\delta)}{\eps^2}
    \cdot (t/\ot)
    \right)$.
    
\end{proof}
Now, we are ready to bound the expected run-time of \ApproxTriangles{}. We consider it step-by step as follows:
\begin{itemize}
    \item Steps 1, 2 and 3 take $O(|R|)$ time.
    \item By Claim \ref{claim: isHeavy gives good partition}, we know that the run-time of \isHeavy{} on an edge $e$ is at most $O\left(\frac{d(e)}{\tau_t} \ln\frac{m}{\delta}\right)$. Thus, running this procedure on every edge in $R$ takes time $O\left(\frac{d(R)}{\tau_t} \ln\frac{m}{\delta}\right)$. Note that to reach Step (4), the algorithm has to avoid terminating at step (3) due to the set $R$ satisfying $d(R)\leq |R|\cdot\ta\cdot \frac{4}{\delta}$. Therefore, we can bound the run-time in step (4) by 
    \[
    O\left(\left(\frac{|R|\cdot \ta}{\tau_t \delta}+|R| \right)\ln\frac{m}{\delta}\right)
    =O\left(\frac{|R|}{\delta}\ln\frac{m}{\delta}\right)
    \]
    \item Step (5) takes time $O(|R|))$.
    \item Setting  aside the invocations of IsAsigned, Step (6) takes time 
    \[O(s)=O\left(\frac{d(R)}{|R|\cdot \frac{\ot}{m}}\cdot\frac{\ln(1/\delta)}{\eps^2}\right) 
    =
    O\left(\frac{\ta\cdot m}{\delta\cdot \ot}\cdot\frac{\ln(1/\delta)}{\eps^2}\right) 
    \]
    \item By Claim \ref{claim: how many times IsAssigned invoked}, the expected number of times that the  subroutine \textsf{IsAsigned} is invoked is at most $O\left(\frac{\ln(1/\delta)}{ \eps^2}\cdot t/\ot\right)$. From Claim \ref{claim: isHeavy gives good partition}, we know that the run-time of \textsf{IsAsigned} on any edge $e$ is at most $O\left(\frac{\tau_d}{\tau_t} \ln\frac{m}{\delta}\right)$. Therefore, in expectation, the total run-time for the invocations of \textsf{IsAsigned} in Step (6) is at most \[
    O\left(\frac{\tau_d}{\tau_t} \cdot \ln\frac{m}{\delta}
    \frac{\ln(1/\delta)}{ \eps^2}\cdot t/\ot
    \right)
    =
    O\left(\frac{m \gamma }{\ot} \cdot \ln\frac{m}{\delta}
    \frac{\ln(1/\delta)}{ \eps^2}\cdot t/\ot
    \right)    
    \]
\end{itemize}
Summing up all the contributions, we see that the overall expected run-time is at most
\[
O\left(\frac{|R|}{\delta}\ln\frac{m}{\delta}\right)
+ O\left(\frac{\ta\cdot m}{\delta\cdot \ot}\cdot\frac{\ln(1/\delta)}{\eps^2}\right)  
+O\left(\frac{m \gamma }{\ot} \cdot \ln\frac{m}{\delta}
    \frac{\ln(1/\delta)}{ \eps^2}\cdot t/\ot
    \right) 
\]
We can bound $R$ as follows:
\[  
    |R|\leq
     \setq+\setQ
     =
     O\left(\frac{m\cdot\gamma}{ \ot}\cdot \frac{\ln \frac{1}{\delta} }{\eps^3}+\frac{m}{(\eps \ot)^{2/3}}\cdot\ln(1/\delta)\right)
     =O\left(\frac{m\cdot\gamma}{ \ot}\cdot \frac{1 }{\eps^3}\ln \frac{1}{\delta}\right)
     .
\]
Substituting the above bound for $|R|$ and $\gamma=\max\{\ta, \ot^{1/3}\}$, our bound on the overall run-time becomes
\begin{multline*}
 O\left(\frac{m\cdot\max\{\ta, \ot^{1/3}\}}{ \ot}\cdot \frac{\ln 1/\delta }{\eps^3 \delta}\ln\frac{m}{\delta}
 +
 \frac{m \max\{\ta, \ot^{1/3}\} }{\ot} \cdot \ln\frac{m}{\delta}
    \frac{\ln(1/\delta)}{ \eps^2}\cdot t/\ot
 \right)=\\
 O\left(\frac{m\cdot\max\{\ta, \ot^{1/3}\}}{ \ot}\cdot \frac{\ln 1/\delta }{\eps^3 \delta}\ln\frac{m}{\delta}\left(1+t/\ot\right)
 \right),
\end{multline*}
finishing the proof of the run-time bound in Theorem \ref{thm: main theorem before search for ot}.

\subsection{Searching for the ``right'' $\ot$ value for Triangle Counting}
\label{sec: Search Theorem for Triangle Counting}

We use to the following as the   ``search theorem".

\begin{theorem}[Theorem 18 in~\cite{ERS_cliques}] \label{thm:search}
	Let  $\invokeA\left(\ov,\eps,\delta, \vecV \right)$ be an algorithm that is given parameters $\ov,\eps, \delta$, possibly an additional set of parameters denoted $\vecV$.

\begin{enumerate}
	\item\label{it:oa-good} If $\ov \in [\vvv/4,\vvv]$, then with probability at least $1-\delta$, $\mA$ returns a value $\hv$ such that
	$\hv\in(1\pm \eps)\vvv$.
	\item\label{it:oa-big} If $\ov > \vvv$, then $\mA$ returns a value $\hv$, such that with probability at least $\eps/4$, $\hv \leq (1+\eps)\vvv$.
	\item\label{it:exp-t-A} The expected running time of $\mA$, denoted $\ert{\mA\left(\ov,\eps,\delta, \vecV\right)}$,
 is monotonically non-increasing with $\ov$ and furthermore, if $\ov < \vvv$, then 
$\ert{\mA\left(\ov,\eps,\delta, \vecV\right)} \leq \ert{\mA\left(\vvv,\eps,\delta, \vecV\right)}\cdot (\vvv/\ov)^\ell$
for some constant $\ell > 0$.
\end{enumerate}
Then there exists an algorithm \search\ that, when given an upper bound $U$ on 
$\vvv$, a parameter $\eps$, a set of parameters $\vecV$ and access to $\mA$, returns a value $X$ such that the following holds. 
\begin{enumerate}
	\item\label{it:Ap-good}  \search$(\mA,U,\eps,\vecV)$  returns a value $X$ such that $X \in (1\pm\eps)\vvv$ with probability at least $4/5$.
	\item\label{it:exp-t-Ap} The expected running time of \search$\left(\mA,U,\eps,\vecV\right)$ is $\ert{\mA\left(\vvv,\eps,\delta, \vecV\right)} \cdot \poly( \log(U), 1/\eps, \ell)$ for $\delta =\Theta\left(\frac{\eps}{2^{\ell}(\ell + \log \log (U))}\right) .$ Moreover, \search$\left(\mA,U,\eps,\vecV\right)$ makes calls to $\mA$ as a black box, and if $N$ denotes the number of calls and $\{\delta_1,\cdots \delta_N\}$ denote the respective value of the failure-probability parameter for each call then
    \begin{equation}
        \label{eq: bound on sum-of-failures for search theorem}
        \sum_{i=1}^N \delta_i\leq \frac{1}{5}
    \end{equation}  
\end{enumerate}
\end{theorem}
\subsection{Using the Search Theorem for testable triangle counting}
In this section we combine Theorem \ref{thm:search} with Theorem \ref{thm: main theorem before search for ot} to obtain the Theorem~\ref{thm: main theorem after search for ot}
Let $G$ be the input
graph, $\eps$ be the approximation parameter, and a failure probability parameter $\delta$.
Let $t$ be the number of triangles in $G$. In addition to augmented query
access to $G$ and approximation parameter $\epsilon$,
the algorithm takes as an input a pair of positive integers $\ta$.
Then, there exists an 
algorithm  with an expected run-time and query complexity of $O^*\left(\left(\frac{m \ta }{ t}+\frac{m }{ t^{2/3}} \right)\log \frac{1}{\delta}\right)$. Every time the algorithm is run, it can either return \badadvice{} or a number $\hat{t}$, satisfying
\begin{itemize}
\item  (Completeness)
For every $G$, if $\aG\leq  \ta$,
then the algorithm \ApproxTriangles{} can output $\badadvice$ only with probability at most $\delta$.

\item  (Soundness)
For every $G$ and $
\ta$,  w.p. at least $1-\delta$ \ApproxTriangles{} will either return $\badadvice$ or output $\hat{t} \in (1 \pm \epsilon)t$.
\end{itemize}

It will suffice to show the theorem above for $\delta=1/5$, after which the case of general $\delta$ will follow by a standard repetition argument. Specifically, this can be achieved by running the algorithm $k=20\log \frac{1}{\delta}$ times and (a) if the majority of runs return $\badadvice$ we likewise return $\badadvice$ (b) otherwise, we return the median of the returned values $\hat{t}_i$. The Completeness condition of the new algorithm follows because if $\aG\leq \ta$ then each of the iterations has only a probability at most $1/5$ of returning $\badadvice$, and therefore by Chernoff Bound the probability that more than half of the iterations return $\badadvice$ is at most $\delta$. The Soundness is true because, again by a Chernoff bound, with probability at least $1-\delta$ the number of iterations that return $\hat{t}_i$ with $\hat{t}_i\notin  [(1-\eps)t, (1+\eps)t]$ is at most $k/5$. Conditioned on this event, if more than half of the iterations produced a value of $\hat{t}_i$ then at least $3/5$ of these values are in the interval $[(1-\eps)t, (1+\eps)t]$ and thus the median of these values is also in the interval $[(1-\eps)t, (1+\eps)t]$.

 To do this, we set
\[
\invokeA\left(\ot,\eps, \delta, \ta \right)
\eqdef
\textsf{Triangle-count-with-advice}(\ot,\eps/20,\delta, \ta)
\]
and define an algorithm $\mathcal{B}(\eps,\ta )$ as follows:
\begin{itemize}
    \item First, we start running  \search$\left(\mA,m^{3/2},\eps,\ta \right)$ with the choice of $\invokeA$ defined above.
    \item If at any point while running \search$\left(\mA,m^{3/2},\eps,\ta \right)$ a call to $\invokeA$ returns \badadvice, we output \badadvice{} and terminate. 
\end{itemize}
Here we take $U=m^{3/2}$ because a graph with $m$ edges has at most $m^{3/2}$ triangles.

\paragraph{Run-time.} By Theorem \ref{thm: main theorem before search for ot}, the algorithm $\mA$ has an expected run-time of $O\left(\frac{m\cdot\max\{\ta, \ot^{1/3}\}}{ \ot}\cdot \frac{\ln \frac{1}{\delta} }{\eps^3\delta} \ln\frac{m}{\delta}
 \cdot (1+t/\ot)
 \right)$. This run-time bound satisfies Premise (3) of Theorem \ref{thm:search} with $\ell=2$. Thus, the expected run-time of  $\mathcal{B}(\eps,\ta )$  is  $O\left(\frac{m\cdot\max\{\ta, t^{1/3}\}}{ t}\cdot \poly\left(\frac{1}{\epsilon}\right)
 \right)=O\left(
 \left(\frac{m \ta }{ t}+\frac{m\cdot }{ t^{2/3}}\right)\cdot \poly\left(\frac{1}{\epsilon}\right)
 \right)$.

 \paragraph{Completeness.} Suppose $\ta\geq \oa$. We now show that $\mathcal{B}(\eps,\ta )$  outputs $\badadvice$ with probability at most $1/5$.

 If  $N$ denotes the number of calls that \search$\left(\mA,m^{3/2},\eps,\ta \right)$ makes to $\mA$, and $\{\delta_1,\cdots \delta_N\}$ and denote the respective value of the failure-probability parameter for each call, then
 Theorem \ref{thm:search} tells us that  
    $
        \sum_{i=1}^N \delta_i\leq \frac{1}{5}
    $.

    When \search$\left(\mA,m^{3/2},\eps,\ta \right)$ calls $\mA$ for the $i\th$ time it uses the failure probability parameter $\delta=\delta_i$. The probability that this $i\th$ call to  $\mA$ returns $\badadvice$ is at most $\delta_i$ by Theorem \ref{thm: main theorem before search for ot}, because  
    since $\ta\geq \oa$. A union bound tells us that with probability at least $1-\sum_{i=1}^N \delta_i\geq \frac{4}{5}$ none of this calls returns $\badadvice$. In this case, $\mathcal{B}(\eps,\ta )$ also will not output $\badadvice$.

\paragraph{Soundness.}
Now, we would like to show that for any value of $\ta$, the algorithm $\mathcal{B}(\eps,\ta )$ will with probability at least $4/5$ either output $\badadvice$ or return a value $\hat{t}$ satisfying $\hat{t}=(1\pm \eps)t$.

Define $\mA'$ to be an oracle that: (1) runs $\mA$ (2) if $\mA$ returns a value $\hat{t}$ then $\mA'$ returns the same value  $\hat{t}$ (3) if $\mA$ returns $\badadvice$ then $\mA'$ returns $t$. 
Similar to $\mA'$, also define an oracle $\mathcal{B'}(\eps,\ta )=$\search$\left(\mA',U,\eps,\vecV\right)$. 
From  Theorem \ref{thm: main theorem before search for ot} we know that 
\begin{itemize}
    \item  
For every $G$ and $
\ta$, if
$\ot \in [  t/4, t ]$ 
then  w.p. at least $1-\delta$, the algorithm $\invokeA\left(\ot,\eps, \delta, \ta \right)$ will either return $\badadvice$ or output $\hat{t} \in (1 \pm \epsilon)t$.

\item  
For every $G$ and $
\ta$, if
$\ot > t $ 
then w.p. at least $\epsilon/4$, the algorithm $\invokeA\left(\ot,\eps, \delta, \ta \right)$  will either return $\badadvice$ or output a value of $\hat{t}$ such that $\hat{t} < (1+\epsilon) t$.
\end{itemize}
Combining this with the definition of the oracle $\mA'$ we see that 
\begin{itemize}
    \item  
For every $G$ and $
\ta$, if
$\ot \in [  t/4, t ]$ 
then  w.p. at least $1-\delta$, the algorithm $\invokeA'\left(\ot,\eps, \delta, \ta \right)$ will output $\hat{t} \in (1 \pm \epsilon)t$.

\item  
For every $G$ and $
\ta$, if
$\ot > t $ 
then w.p. at least $\epsilon/4$, the algorithm $\invokeA'\left(\ot,\eps, \delta, \ta \right)$  will output a value of $\hat{t}$ such that $\hat{t} < (1+\epsilon) t$.
\end{itemize}
Thus, Theorem \ref{thm:search} tells us that with probability at least $4/5$ the oracle  $\mathcal{B'}(\eps,\ta )=$\search$\left(\mA',U,\eps,\vecV\right)$ returns a value $\hat{t}\in (1\pm \eps)t$. 

Comparing $\mathcal{B'}(\eps,\ta )=$\search$\left(\mA',U,\eps,\vecV\right)$ with $\mathcal{B}(\eps,\ta )$, we see that the algorithm $\mathcal{B}(\eps,\ta )$ either outputs $\mathcal{B'}(\eps,\ta )$ or it outputs $\badadvice$ (when $\mathcal{B}$ and $\mathcal{B}'$ are instantiated using the same values of random coins). Since with probability at least $4/5$ the oracle  $\mathcal{B'}(\eps,\ta )$ returns a value $\hat{t}\in (1\pm \eps)t$, we conclude that with probability at least $4/5$ the algorithm $\mathcal{B}(\eps,\ta )$ either outputs $\badadvice$ or returns a value $\hat{t}\in (1\pm \eps)t$.

\end{document}